\documentclass{article}%
\usepackage{amssymb}
\usepackage{amsfonts}
\usepackage{amsmath}
\usepackage{algorithmic}
\usepackage{float}
\usepackage{graphicx}
\usepackage{wasysym}%
\usepackage[font=small,format=plain,labelfont=bf,up,textfont=it,up]{caption}
\usepackage[caption=false]{subfig}
\usepackage{morefloats}
\usepackage{epstopdf}
\usepackage{wrapfig}
\providecommand{\U}[1]{\protect\rule{.1in}{.1in}}
\renewcommand{\baselinestretch}{1.2}
\newtheorem{theorem}{Theorem}

\newtheorem{lemma}[theorem]{Lemma}

\newenvironment{proof}[1][Proof]{\noindent\textbf{#1.} }{\ \rule{0.5em}{0.5em}}
\begin{document}

\title{Correlation between centrality metrics and their application to the opinion model}

\author{Cong Li$^{1}$, Qian Li$^{2}$, Piet Van Mieghem$^{1}$, H. Eugene Stanley$^{2}$,
Huijuan Wang$^{1,2}$}

\date{$^{1}$Faculty of Electrical Engineering, Mathematics and Computer Science, \\
Delft University of Technology, Delft, The Netherlands\\
$^{2}$Center for Polymer Studies, Department of Physics, \\
Boston University, Boston, Massachusetts 02215, USA\\}

\maketitle

\begin{abstract}

In recent decades, a number of centrality metrics describing network
properties of nodes have been proposed to rank the importance of nodes.
In order to understand the correlations between centrality metrics and
to approximate a high-complexity centrality metric by a strongly
correlated low-complexity metric, we first study the correlation between
centrality metrics in terms of their Pearson correlation coefficient and
their similarity in ranking of nodes. In addition to considering the
widely used centrality metrics, we introduce a new centrality measure,
the degree mass. The $m$th-order degree mass of a node is the sum of the
weighted degree of the node and its neighbors no further than $m$ hops
away. We find that the betweenness, the closeness, and the components of
the principal eigenvector are strongly correlated with the degree, the
$1$st-order degree mass and the $2$nd-order degree mass, respectively,
in both network models and real-world networks. We then theoretically
prove that the Pearson correlation coefficient between the principal
eigenvector and the $2$nd-order degree mass is larger than that between
the principal eigenvector and a lower order degree mass. Finally, we
investigate the effect of the inflexible contrarians selected based on
different centrality metrics in helping one opinion to compete with
another in the inflexible contrarian opinion (ICO) model. Interestingly,
we find that selecting the inflexible contrarians based on the leverage,
the betweenness, or the degree is more effective in opinion-competition
than using other centrality metrics in all types of networks. This
observation is supported by our previous observations, i.e., that there
is a strong linear correlation between the degree and the betweenness,
as well as a high centrality similarity between the leverage and the
degree.

\end{abstract}

\section{Introduction}

Recent research has explored social dynamics
\cite{strogatz2001exploring, boccaletti2006complex, barrat2008dynamical}
by using complex networks in which nodes represent people/agents and
links the associations between them. Such centrality metrics as degree
and betweenness have been studied in dynamic processes
\cite{comin2011identifying, kitsak2010identification, borge2012absence,
  pastor2002immunization}, such as opinion competition, epidemic
spreading, and rumor propagation on complex networks. These studies used
centrality metrics to identify influential nodes
\cite{comin2011identifying, kitsak2010identification, borge2012absence},
such as the source nodes from which a virus spreads and the nodes with
high spreading capacity, as well as to select which nodes are to be
immunized when a virus is prevalent
\cite{pastor2002immunization}. Numerous centrality metrics have been
proposed. Degree, betweenness, closeness, and principal eigenvector are
the most popular centrality metrics \cite{comin2011identifying,
  borgatti2005centrality, freeman1979centrality,
  Friedkin1991theoretical, mullen1991effects, newman2008mathematics,
  van2014graph}. Several new centrality metrics have been introduced in
a number of different fields recently. Kitsak \textit{et al.}
\cite{kitsak2010identification} studied the SIS and SIR spreading models
on four real-world networks and proposed that the $k$-shell index is a
better indicator for the most efficient spreaders (nodes) than degree or
betweenness. Reference~\cite{joyce2010new} proposes a new centrality
metric---{\it leverage}---for identifying neighborhood hubs (the most
highly-connected nodes) in functional brain networks. Leverage
centrality identifies nodes that are connected to more nodes than their
nearest neighbors. In addition to considering these widely-used
centrality metrics, we here propose a new centrality metric, {\it degree
  mass}. The $m$th-order degree mass of a node is defined as the sum of
the weighted degree of its $m$-hop neighborhood\footnote{The $m$-hop
  neighborhood of a node $i$ includes the node $i$ and all nodes no
  further away than $m$ hops from $i$.\label{footnote1}}. If the degree
of a node and of its neighbors are all high, the node has a high degree
mass.

Centrality metrics have been compared in various networks, such as
sampled networks, biological networks, food webs, and vocabulary
networks in literature \cite{comin2011identifying, kim2007reliability,
  koschutzki2004comparison, estrada2007characterization,
  li2012degree}. Comin \textit{et al.} \cite{comin2011identifying}
compared the centrality metrics characterizing the performances of nodes
in such dynamic processes as virus spreading. Kim and Jeong
\cite{kim2007reliability} compared the reliability of rank orders using
centrality metrics in sampling networks. The correlations between
centrality metrics have been studied in biological networks
\cite{koschutzki2004comparison, estrada2007characterization}.  However
correlations between centrality metrics are still not well
understood. If correlations between centrality metrics were better
understood, we might be able to rank the nodes in a network by using the
centrality metrics with a low computational complexity instead of the
ones with a high computational complexity. To investigate the
correlation between any two centrality metrics, we compute their Pearson
correlation coefficient and their similarity in ranking nodes in both
network models and real-world networks. In this work (i) we consider
Erd\H{o}s-R\'{e}nyi (ER) networks\footnote{An Erd\H{o}s-R\'{e}nyi random
  graph $G_{p}(N)$ can be generated from a set of $N$ nodes by randomly
  assigning a link with probability $p$ to each pair of nodes.} with a
binomial degree distribution \cite{Erdos1959random} and scale-free (SF)
networks\footnote{A scale-free network is characterized by a power-law
  degree distribution \textrm{Prob}$[D=k]\sim k^{-\alpha}$, with
  $k_{\min}\leq k<k_{\max}$. Here, we choose $k_{\min}=2$, $k_{\max}$ as
  the natural cutoff and $\alpha=2.5$.} with a power-law degree
distribution \cite{barabasi1999emergence, cohen2010complex}. Studying
these two network models allows us to understand how the degree
distribution influences correlations between the centrality
metrics. (ii) We further explore correlations in 34 real-world networks
with differing numbers of nodes and links. (iii) We theoretically
compare the Pearson correlation coefficients between the principal
eigenvector and the degree masses.

Recently there has been considerable interest in understanding how two
competing opinions \cite{galam2005local, castellano2009statistical,
  shao2009dynamic, li2013non, qu2014non} evolve in a population. In this
work we apply our centrality metrics to an inflexible contrarian opinion
(ICO) model \cite{li2011strategy} in which only two opinions (denoted
$A$ and $B$) exist, with the goal of helping one opinion (opinion $B$)
as it competes with with the other opinion (opinion $A$). At the initial
time, opinions are randomly assigned to all nodes (with a fraction $f$
of nodes holding opinion $A$ and a fraction $1-f$ of nodes holding
opinion $B$). At each step, each agent simultaneously and in parallel
adopts the opinion of the majority of its nearest neighbors and itself,
and if there is a tie, the agent does not change its opinion. After the
system reaches a steady state, a fraction $p_{o}$ of agents with opinion
$A$ is placed among the inflexible contrarians permanently holding
opinion $B$, which can affect the opinion of their nearest neighbors. It
is known that the size of the giant component of agents with opinion $A$
can be decreased or even destroyed by the inflexible contrarians
\cite{li2011strategy}. Li \textit{et al.} \cite{li2011strategy} have
selected the inflexible contrarians in ER and SF networks either
randomly or based on degree. Here we choose inflexible contrarians using
all the centrality metrics we have considered in both modelled networks
and real-world networks. We compare the efficiencies of these centrality
metrics in reducing the size of the largest opinion $A$ cluster and find
that strongly correlated centrality metrics have approximately the same
efficiency in both modelled networks and real-world networks. Thus a
high-complexity centrality metric could be approximated by a strongly
correlated low-complexity centrality metric.

This paper is organized as follows. In Sec.~\ref{centraltiy_metrics} we
introduce the centrality metrics. In Sec.~\ref{relation} we study the
Pearson correlation coefficient and the centrality similarity between any two
centrality metrics in both network models and real-world networks. In 
Sec.~\ref{Sec_theoretical_analysis} the Pearson correlations between the degree
masses and the principal eigenvector are theoretically analysed. In 
Sec.~\ref{opinionmodel} the centrality metrics are applied in choosing the
inflexible contrarians in the ICO model and the efficiencies of the
centrality metrics are compared.

\section{Definition of network centrality metrics}
\label{centraltiy_metrics}

Centrality metrics quantify node properties in a network. Here we first
review some centrality metrics that are widely used or have been
recently proposed \cite{comin2011identifying, borgatti2005centrality,
  freeman1979centrality, Friedkin1991theoretical, mullen1991effects,
  newman2008mathematics, kitsak2010identification, joyce2010new,
  van2014performance}. We then propose a new centrality metric, which we
call {\it degree mass}. Let $G(\mathcal{N}$, $\mathcal{L})$ be a
network, where $\mathcal{N}$ is the set of nodes and $\mathcal{L}$ is
the set of links. The number of nodes is denoted by $N=|\mathcal{N}|$
and the number of links by $L=|\mathcal{L}|$. The network $G$ can be
represented by an $N\times N$ symmetric adjacency matrix $A$, consisting
of elements $a_{ij}$, which are either one or zero depending on whether
node $i$ is connected to node $j$ or not. The networks mentioned in this
paper are simple, unweighted and do not have self-loops or multiple
links.

\begin{itemize}
\item Principal eigenvector $x_{1}$
\end{itemize}

The largest eigenvalue of the adjacency matrix $A$ is $\lambda_{1}$,
also called the spectral radius \cite{VanMieghem2011Graph_spectra}. The
principal eigenvector $x_{1}$\ corresponding to the spectral radius
$\lambda_{1}$ satisfies the eigenvalue equation
\[
Ax_{1}=\lambda_{1}x_{1}.
\]
Component $j$ of the principal eigenvector is denoted by
$(x_{1})_{j}$. The $X_{1}$ is the element in the principal eigenvector
that corresponds to a random node.

\begin{itemize}
\item Betweenness $B_{n}$
\end{itemize}

Betweenness was introduced independently by Anthonisse
\cite{anthonisse1971rush} in 1971 and Freeman
\cite{freeman1979centrality} in 1977. The betweenness of a node $i$ is
the number of shortest paths between all possible pairs of nodes in the
network that traverse the node
\[
b_{ni}=\sum_{s\neq i\neq d\in\mathcal{N}}\frac{\sigma_{sd}(i)}{\sigma_{sd}},
\]
where $\sigma_{sd}(i)$ is the number of shortest paths that pass through
node $i$ from node $s$ to node $d$, and $\sigma_{sd}$ is the total
number of shortest paths from node $s$ to node $d$. The betweenness
$B_{n}$ incorporates global information and is a simplified quantity for
assessing the traffic carried by a node. Assuming that a unit packet is
transmitted between each node pair, the betweenness $b_{ni}$ is the
total number of packets passing through node $i$
\cite{wang2008betweenness}.

\begin{itemize}
\item Closeness $C_{n}$
\end{itemize}

The closeness \cite{koschutzki2005centrality} of a node $i$ is the
average hopcount of the shortest paths from node $i$ to all other
nodes. It measures how close a node is to all the others. The most
commonly used definition is the reciprocal of the total hopcount,
\[
c_{ni}=\frac{N-1}{\sum_{j\in\mathcal{N}\setminus\{i\}}H_{ij}},
\]
where $H_{ij}$ is the hopcount of the shortest path between nodes $i$
and $j$, and $\sum_{j\in\mathcal{N}\setminus\{i\}}H_{ij}$ is the sum of
the hopcount of the shortest paths from node $i$ to all other
nodes. Closeness has been used to identify central metabolites in
metabolic networks \cite{ma2003connectivity}.

\begin{itemize}
\item $K$-shell index $K_{s}$
\end{itemize}

The $k$-shell decomposition of a network allows us to identify the core
and the periphery of the network. The $k$-shell decomposition proceedure
is as follows:

\begin{itemize}

\item[{(1)}] Remove all nodes of degree $d=1$ and also their links. This
  may reduce the degree of other nodes to 1.

\item[{(2)}] Remove nodes whose degree has been reduced to 1 and their
  links until all of the remaining nodes have a degree $d>1$. All of the
  removed nodes and the links between them constitute the $k$-shell with
  an index $k_{s}=1$.

\item[{(3)}] Remove nodes with degree $d=2$ and their links in the
  remaining networks until all of the remaining nodes have a degree
  $d>2$. The newly removed nodes and the links between them constitute
  the k-shell with an index $k_{s}=2$, and subsequently for higher
  values of $k_{s}$.

\end{itemize}

\noindent
The $k$-shell is a variant of the $k$-core \cite{seidman1983network,
  pittel1996sudden}, which is the largest subgraph with minimum degree
of at least $k$. A $k$-core includes all $k$-shells with an index of
$k_{s}=0,1,2,\cdots, k$. An $O(m)$ algorithm for $k$-shell network
decomposition was proposed in Ref.~\cite{batagelj2003m}. The $k$-shell
index of the original infected node is a better predictor of the
infected population in the susceptible-infectious-recovered (SIR)
epidemic spreading process than other centrality metrics, such as the
degree \cite{kitsak2010identification}.

\begin{itemize}
\item Leverage $L_{n}$
\end{itemize}

Joyce \textit{et al.} \cite{joyce2010new} introduced leverage centrality
in order to identify neighborhood hubs in functional brain networks. The
leverage measures the extent of the connectivity of a node relative to
the connectivity of its nearest neighbors. The leverage of a node $i$ is
defined
\[
l_{ni}=\frac{1}{d_{i}}\sum_{j\in\mathcal{N}_{i}}\frac{d_{i}-d_{j}}{d_{i}
+d_{j}},
\]
where $\mathcal{N}_{i}$ is the directly connected neighbors of the node
$i$. With the definition of $l_{ni}$ and the range $[1,N-1]$ of the
degree $d_{i}$ in connected networks, the leverage of a node $i$ is
bounded by $-1+\frac{2d_{i}}{d_{i}+(N-1)}\leq
l_{ni}\leq1-\frac{2}{d_{i}+1}$. Hence the range of the leverage $l_{ni}$
is $[-1+2/N, 1-2/N]$ and the equality occurs in star graphs and complete
graphs $K_{N}$. The leverage of a node is high when it has more
connections than its direct neighbors. Thus a high-degree node with
high-degree nearest neighbors will probably have a low leverage.

\begin{itemize}
\item Degree mass $D^{(m)}$
\end{itemize}

The degree of a node $i$ in a network $G$ is the number of its direct
neighbors,
\[
d_{i}=\sum_{j=1}^{N}a_{ij}=(Au)_{i},
\]
where $u=(1,1,\cdots,1)^{T}$ is the all-one vector. Here we propose a
new set of centrality metrics, the degree mass, which is a variant of
degree centrality. The $m$th-order degree mass of a node $i$ is defined
as the sum of the weighted degree of its $m$-hop
neighborhood,
\[
d_{i}^{(m)}=\sum_{k=1}^{m+1}\left(  A^{k}u\right)_{i}=
\sum_{j=1}^{N}\left(\sum_{k=0}^{m}A^{k}\right)_{ij}d_{j}, 
\]
where $m\geq0$. The weight of the degree $d_{j}$ is the number
of walks\footnote{A walk from $i$ to $j$ is any sequence of edges that
  allows back and forth movement and repeated visits to the same node.}
of length no longer than $m$ from node $i$ to node $j$. The weight of
$d_{j}$ is larger than the weight of $d_{l}$ when node $l$ is farther
than node $j$ from node $i$. The $m$th-order degree mass vector is
defined $d^{(m)}=[d_{1}^{(m)},d_{2}^{(m)},\cdots,d_{N}^{(m)}]$. The
0th-order degree mass is the degree centrality. The 1st-order degree
mass of node $i$ is the sum of the degree of node $i$ and the degree of
its nearest neighbors. When $m$ is large, the $m$th-order degree mass is
proportional to the principal eigenvector.

\section{Correlations between centrality metrics}
\label{relation}

We investigate the correlations between the centrality metrics
introduced in Sec.~\ref{centraltiy_metrics}, in both network models and
real-world networks.  The network models include the Erd\H{o}s-R\'{e}nyi
(ER) network and the scale-free (SF) network. ER networks are
characterized by a binomial degree distribution with
$\mathrm{Prob}\left[D=k\right]=\binom{N-1}{k}p^k(1-p)^{N-1-k}$, where
$N$ is the number of nodes and $p$ is the probability that each node
pair is connected. A SF network \cite{barabasi1999emergence,
  cohen2000resilience} has a power-law degree distribution with
\textrm{Prob}$[D=k]\sim k^{-\alpha}$, $k\in \left[k_{\rm min}, k_{\rm
    max}\right]$, where $k_{\rm min}$ is the smallest degree, $k_{\rm
  max}$ is the degree cutoff, and $\alpha$ is the exponent
characterizing the broadness of the distribution. In this work we use
the natural cutoff at approximately $N^{1/(\alpha-1)}$ and $k_{\rm min}
= 2$. We consider 34 real-world networks, e.g., airline connections,
electrical power grids, and coauthorship collaborations. The
descriptions and properties of these real-world networks are given in
Appendix $A$. We study the correlations between any two centrality
metrics using the Pearson correlation coefficient and the centrality
similarity.

\subsection{Pearson correlation coefficients between centrality metrics}
\label{Linearcorrelation}

Here we explore the linear correlation between the centrality metrics
using numerical simulations in both ER and SF networks as well as in
real-world networks. The results in Appendix $B$ indicate that strong
linear correlations do exist between certain centrality metrics in both
ER and SF networks, and that network size has little influence on the
correlations. Note that the $k$-shell index is weakly correlated with
all the other centrality metrics. This might be the case because the
$k$-shell indices of all nodes are similar to each other in binomial
networks. We note the following seemingly universal relations between
the degree masses and three centrality metrics, the principal
eigenvector $x_{1}$, the closeness $C_{n}$ and the betweenness $B_{n}$,
as
\[
\begin{cases}
\rho(X_{1},D^{(2)})>\rho(X_{1},D^{(1)})>\rho(X_{1},D), \\
\rho(C_{n},D^{(1)})>\rho(C_{n},D^{(2)})>\rho(C_{n},D),\\
\rho(B_{n},D)>\rho(B_{n},D^{(2)})>\rho(B_{n},D^{(1)}),\\
\end{cases}
\]
in most real-world networks (see Figs.~\ref{eig}, \ref{clos}, and
\ref{betw}). The same results can be found in both ER and SF networks
(see Appendix $B$). We theoretically prove the inequality
$\rho(X_{1},D^{(2)})>\rho(X_{1},D^{(1)})>\rho(X_{1},D)$ in ER networks
in Sec.~\ref{Sec_theoretical_analysis}.

\begin{figure}[H]
\centering
\subfloat[]{\label{eig}
\includegraphics[width=0.33\textwidth]{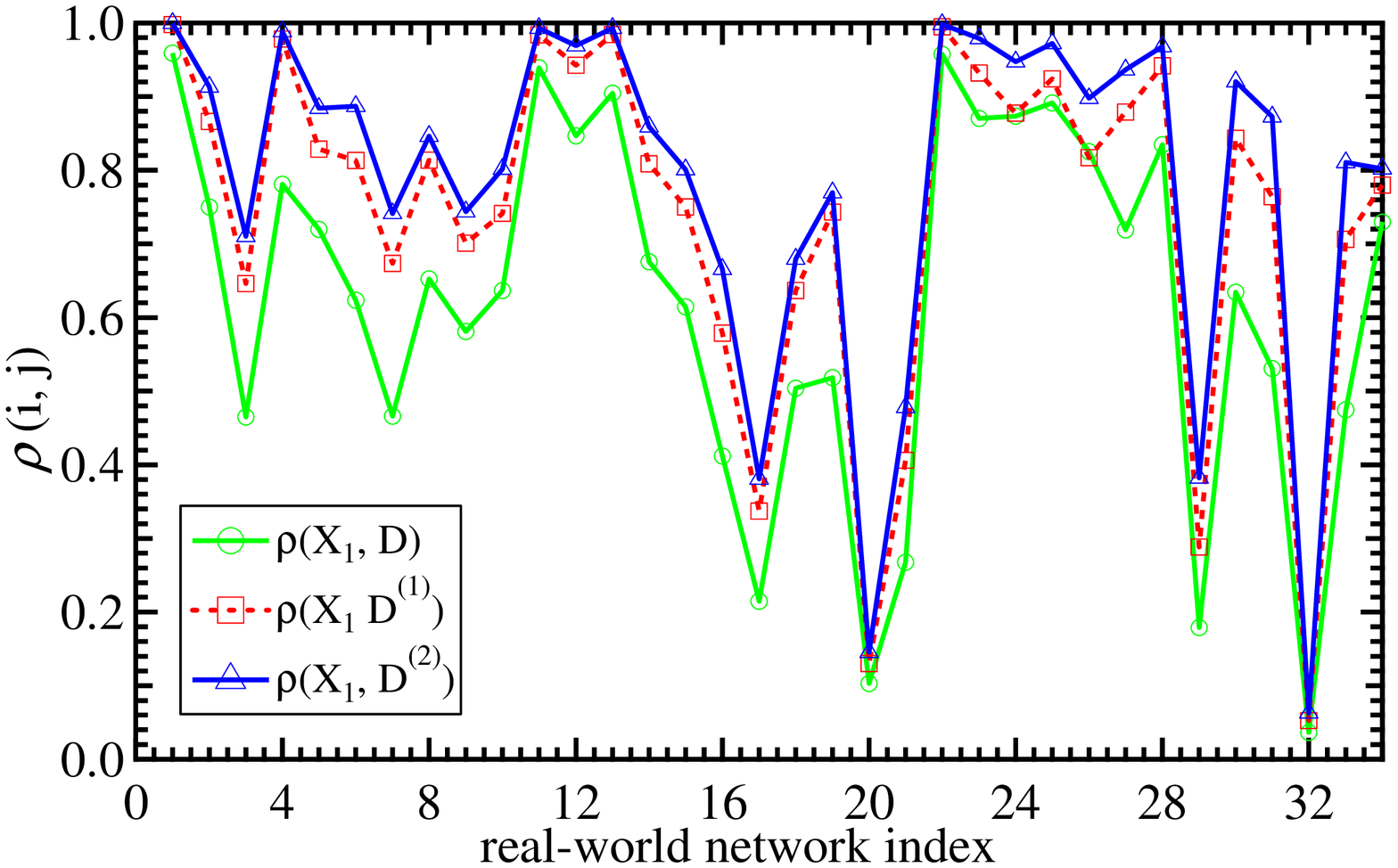}}
\subfloat[]{\label{clos}
\includegraphics[width=0.33\textwidth]{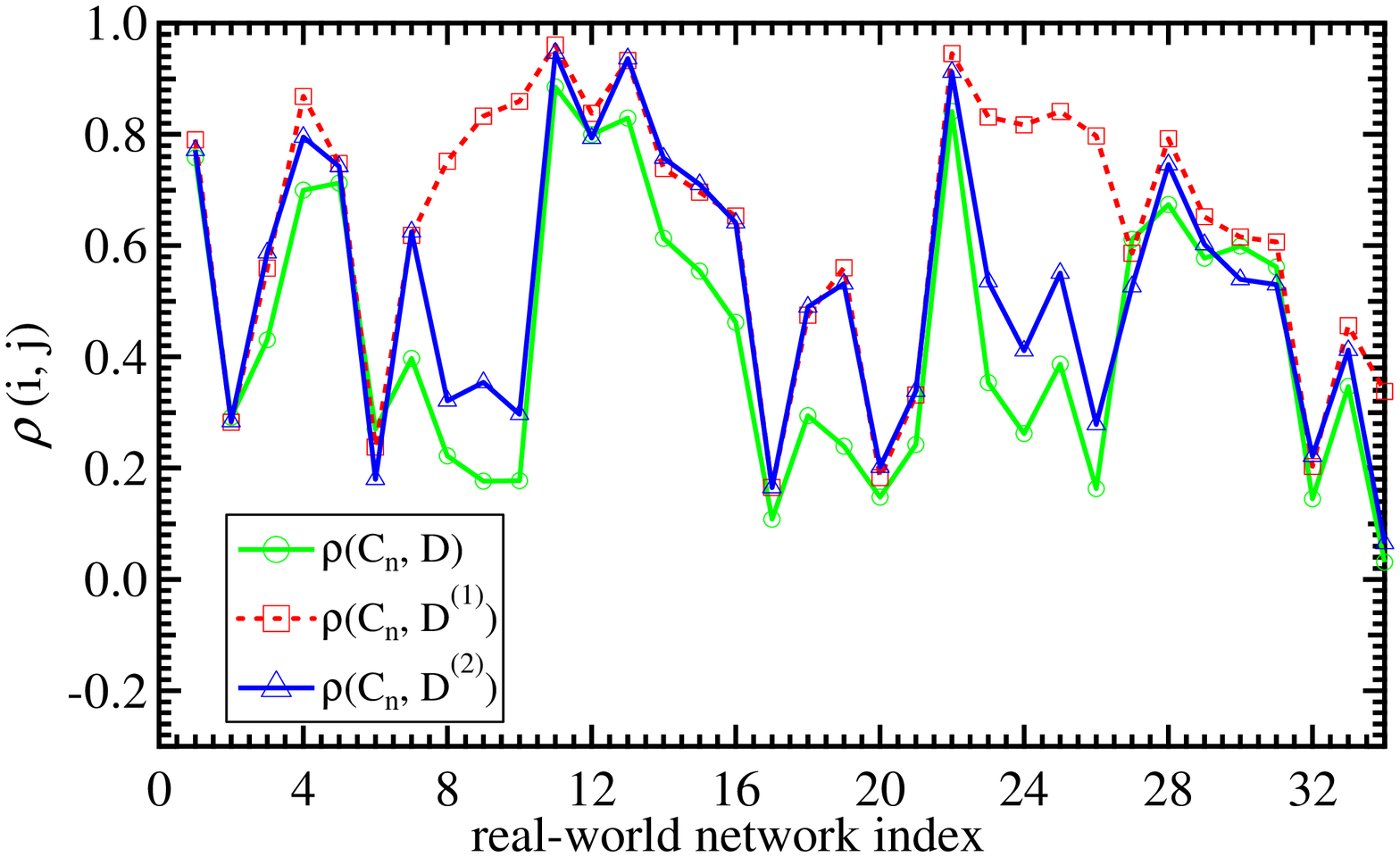}}
\subfloat[]{\label{betw}
\includegraphics[width=0.33\textwidth]{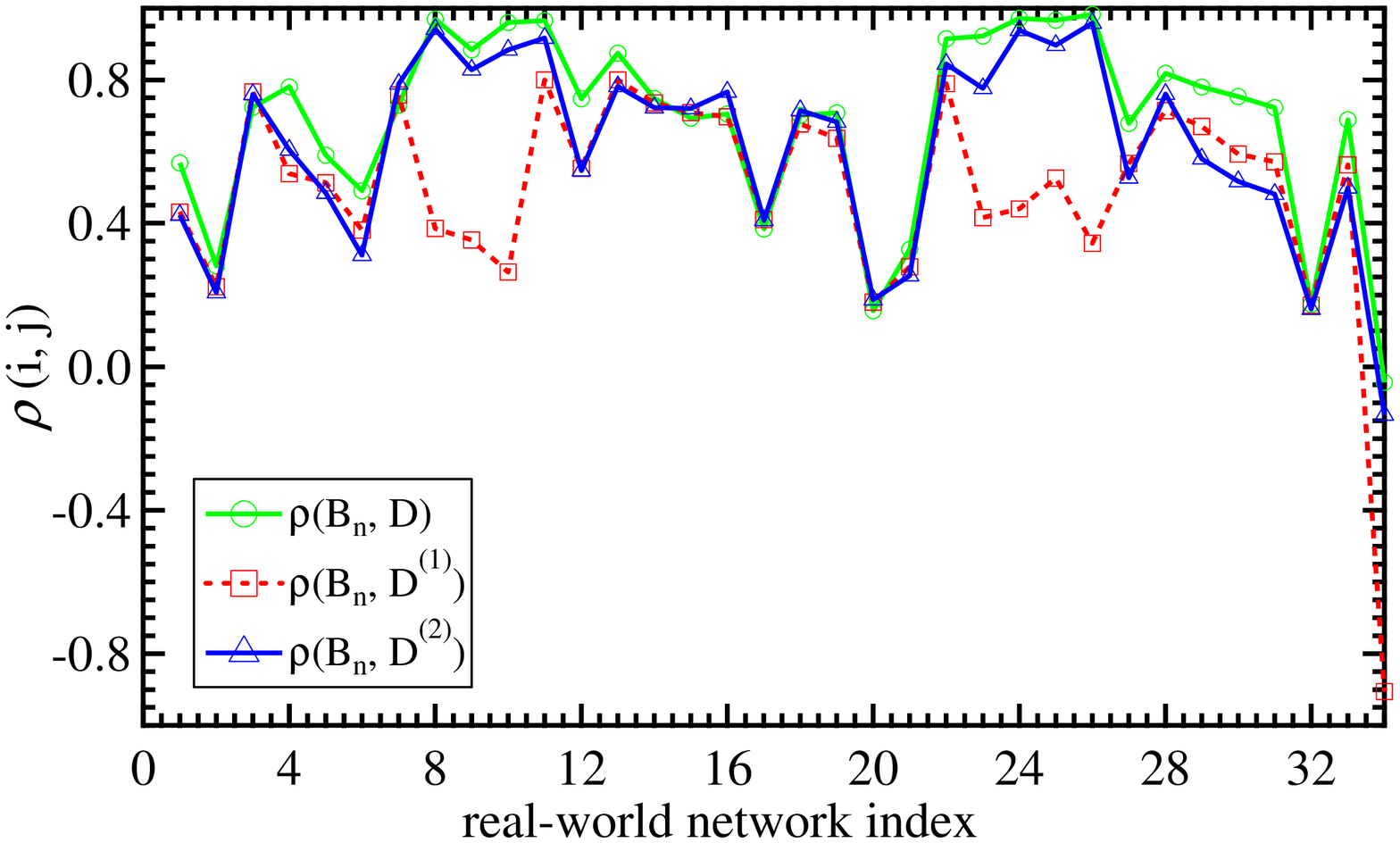}}
\bigskip
\caption{Pearson correlation coefficients (a) between the principal
  eigenvector and the degree masses: $\rho(X_{1},D)$ (in circle marks),
  $\rho(X_{1},D^{(1)})$ (in rectangle marks), and $\rho(X_{1},D^{(2)})$
  (in\ triangle marks); (b) between the closeness and the degree masses:
  $\rho(C_{n},D)$ (in circle marks), $\rho(C_{n},D^{(1)})$ (in rectangle
  marks), and $\rho(C_{n},D^{(2)})$ (in\ triangle marks); (c) between
  betweenness and degree masses: $\rho(B_{n},D)$ (in circle marks),
  $\rho(B_{n},D^{(1)})$ (in rectangle marks), and $\rho(B_{n},D^{(2)})$
  (in\ triangle marks), in 34 real-world networks.}
\label{correlation_figure1}
\end{figure}

Almost all of the Pearson correlation coefficients
$\rho(X_{1},D^{(2)})$, $\rho(C_{n},D^{(1)})$, and $\rho(B_{n},D)$ are
large ($>0.95$) in both ER and SF networks (see Figs.~\ref{ER3} and
\ref{SF}) and are also large ($>0.6$) in most real-world networks (see
Fig.~\ref{correlation_figure1}). The betweenness of a power-law
distributed network also follows a power-law distribution
\cite{joy2005high}. This supports the strong linear correlation between
the betweenness $B_{n}$ and the degree $D$ in SF networks
\cite{estrada2007characterization}.

\subsection{Centrality similarities $M_{A,B}(\Upsilon)$ between centrality metrics}

Different centrality metrics rank the nodes in different orders within a
network. The centrality similarity was proposed in
Ref.~\cite{trajanovski2013robustness} to quantify the similarity of
centrality metrics in ranking nodes.

\textbf{Definition} In a graph $G(N,L)$ assume we obtain two node
rankings, $[a_{(1)},a_{(2)},\cdots,a_{(N)}]$ and
$[b_{(1)},b_{(2)},\cdots,$\ $b_{(N)}]$, according to centrality metrics
$A$ and $B$, where $a_{(j)}$ or $b_{(j)}$ is the node whose centrality
metric $A$ or $B$ is the $j$-th largest in the networks. The centrality
similarity $M_{A,B} (\Upsilon)$ is the percentage of the nodes in
$[a_{(1)} ,a_{(2)},\cdots\cdots,a_{(\Upsilon N)}]$, which are also in
$[b_{(1)},b_{(2)},\cdots\cdots,b_{(\Upsilon N)}]$, where
$\Upsilon\in\lbrack0,1]$.

The measure $M_{A,B}(\Upsilon)$ gives the percentage of overlapping
nodes from the top $100\Upsilon\%$ of nodes, ranked by the centrality
metrics $A$ and $B$, respectively. The range of $M_{A,B}(\Upsilon)$ is
between $[0,1]$. If the $100\Upsilon\%$ of nodes chosen by centrality
metric $A$ are not at all in the $100\Upsilon\%$ of nodes chosen by
centrality metric $B$, $M_{A,B}(\Upsilon)=0$. It means that the most
important (top $100\Upsilon\%$) nodes chosen by the two centrality
metrics are completely different, i.e., the centrality metrics $A$ and
$B$ differ greatly. When all nodes are chosen ($\Upsilon=1$) there is a
full overlap, which indicates that $M_{A,B}(1)=1$. For a given
$\Upsilon<1$, a larger $M_{A,B}(\Upsilon )$ represents a stronger
correlation between the two centrality metrics $A$ and $B$.

\subsubsection{Centrality similarities in network models}

We study the centrality similarity $M_{A,B}(\Upsilon)$ between any two
centrality metrics\footnote{Our study shows that the centrality
  similarity $M_{A,B}(\Upsilon)$ increases with the increase of
  $\Upsilon$ in ER networks, but decreases with the increase of
  $\Upsilon$ in SF networks. Note that this observation holds only for
  small $\Upsilon$ and, if $\Upsilon$ is around $1$,
  $M_{A,B}(\Upsilon)=1$ in all networks.} in $10^{3}$ network
realizations of ER networks and SF networks with $N=10^{4}$ and
$\Upsilon=[0.001$, $0.01$, $0.1]$.

\begin{figure}[h]
\centering

\subfloat[]{
\includegraphics[width=0.65\textwidth]{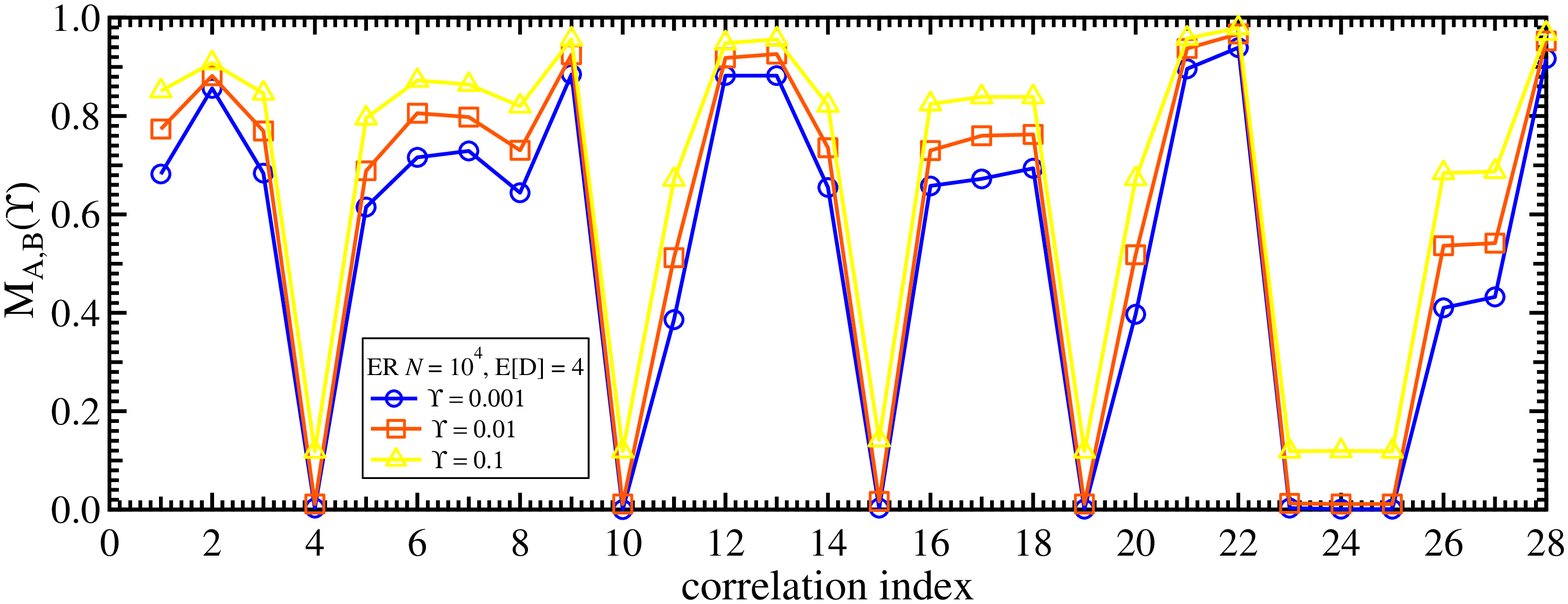}}\\
\subfloat[]{
\includegraphics[width=0.65\textwidth]{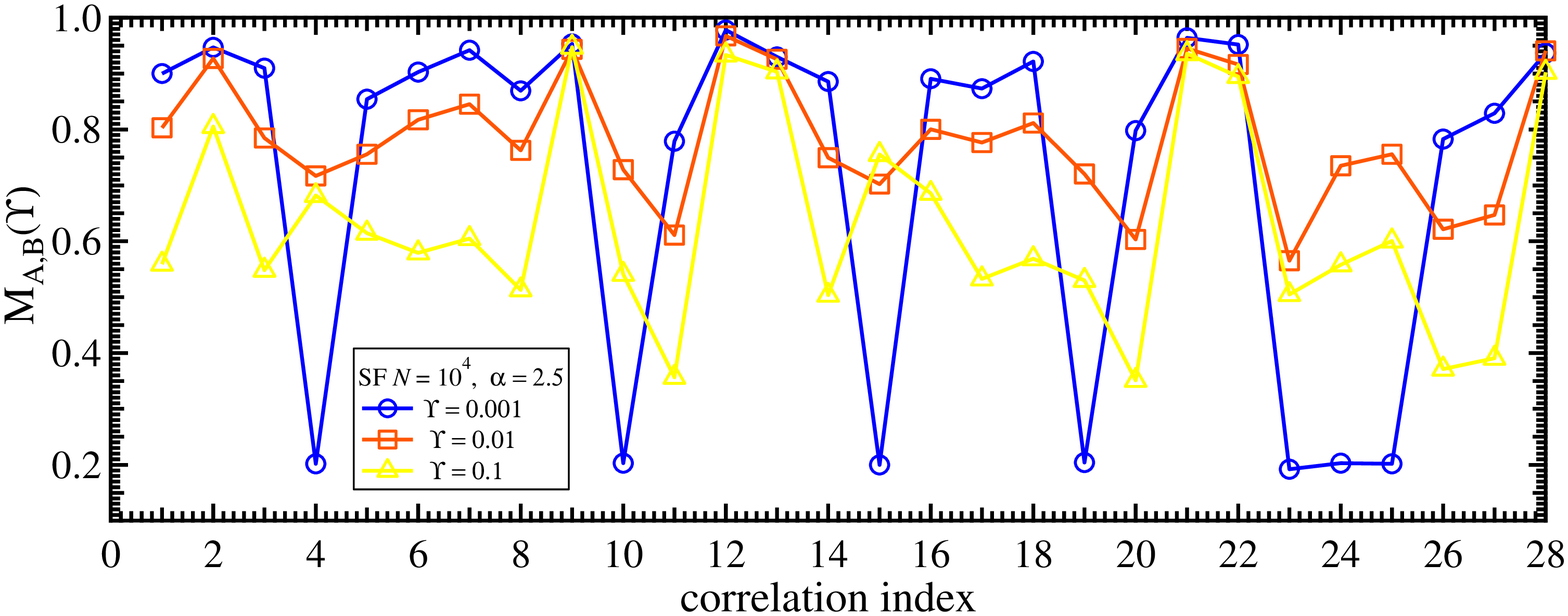}}
\bigskip
\caption{Centrality similarities between centrality metrics in network
  models: (a) for ER networks and (b) for SF networks. The x-axis is the
  correlation index (see Appendix
  \ref{Appen_Pearson}).}\label{opinion_figure0}
\end{figure}

We observe that in both ER and SF networks, the $M_{B_{n},D}(\Upsilon)$
is notably larger than the centrality similarity between $B_{n}$\ and
any other centrality metric;
$M_{C_{n},D^{(1)}}(\Upsilon)>M_{C_{n},D^{(2)}}(\Upsilon)>M_{C_{n},D}(\Upsilon)$;
and the centrality similarities $M_{x_{1},D^{(1)}}(\Upsilon)$ and
$M_{x_{1},D^{(2)}}(\Upsilon)$ are both large. In ER networks,
$M_{x_{1},D^{(2)}}(\Upsilon)>M_{x_{1},D^{(1)}}(\Upsilon)>M_{x_{1},D}(\Upsilon)$. The
$k$-shell index has low similarity with other metrics in ER networks for
the same reason mentioned in Sec.~\ref{Linearcorrelation}. All these
observations agree with what we have found using the Pearson correlation
coefficients in Sec.~\ref{Linearcorrelation}.

\subsubsection{Centrality similarities in real-world networks}

For the 34 real-world networks the percentage $\Upsilon$ should be
larger than 3\%, since the smallest network only has 35 nodes. We
compare the similarity between each centrality metric (e.g., $B_{n}$)
and all other metrics to determine which metric is the closest to the
centrality metric (e.g., $B_{n}$). In Fig.~\ref{barfig} the height of
each bar indicates the number of networks in which $M_{A,B}(\Upsilon)$
is the highest among the centrality similarities between $A$ and all the
other centrality metrics. The bar chart shows that the $D$, $D^{(1)}$,
and $D^{(2)}$ are, respectively, most similar to $B_{n}$, $C_{n}$, and
$x_{1}$ in most real-world networks, which is consistent with what is
observed in the network models. We also observe that either
$M_{L_{n},D}(\Upsilon)$ or\ $M_{L_{n},B_{n}}(\Upsilon )$ is the largest
among the centrality similarities between $L_{n}$ and all other metrics
in most real-world networks.

\begin{figure}[ptb]%
\centering
\includegraphics[height=2.435in,width=4.426in]{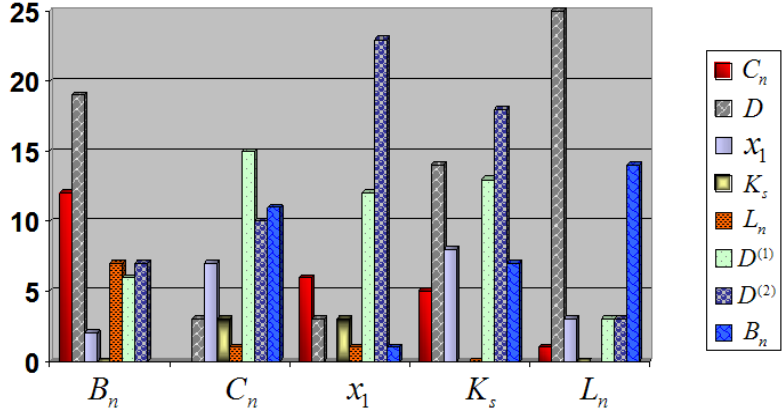}
\caption{Number of networks (among the 34 real-world networks) in which
  $M_{A,B}(\Upsilon)$ is the highest among the centrality similarities
  between $A$ and all other centrality metrics, when $\Upsilon=5\%$. The
  centrality metric $A$ is given by the $x$-axis label, and $B$ is
  reflected by the pattern described in the box on right side. Take the
  betweenness $B_{n}$ as an example. The centrality similarities between
  $B_{n}$\ and all the other metrics are compared with each other to
  find the largest similarity in each real-world network. For instance,
  the $M_{B_{n},C_{n}}(\Upsilon)$ is the largest centrality similarity
  in `Electric\_s208' network, so that one is counted into the leftmost
  bar of $B_{n}$ (with $C_{n}$).}
\label{barfig}
\end{figure}

\section{Theoretical analysis}
\label{Sec_theoretical_analysis}

The above simulations indicate that the three lowest-order degree
masses, with a low computational complexity, are strongly correlated
with the betweenness, the closeness, and the components of the principal
eigenvector, all of which are complex to compute. We first prove that
the high-order ($m\rightarrow\infty$) degree mass is proportional to the
principal eigenvector $x_{1}$ in any network. Next we prove that when
$m$ is small the correlation between degree mass and the principal
eigenvector increases with an increase in $m$, i.e., $\rho(X_{1},
D^{(2)}) \geq \rho(X_{1}, D^{(1)})\geq\rho(X_{1},D)$. We then apply the
generating function method \cite{VanMieghem2006performance,
  newman2001random} to analyze such statistical properties of the degree
masses as expectation and variance (see Appendix~\ref{proof_app}).

\begin{theorem}
\label{lemma_proportional}
The $m$th-order degree mass vector $d^{(m)}$ is proportional to the
principal eigenvector $x_{1}$ in any network with a sufficiently large spectral gap when $m\rightarrow\infty$.
\end{theorem}

\begin{proof}

The $m$th-order degree mass vector $d^{(m)}$ is

\begin{align*}
d^{(m)} & =\sum\limits_{k=1}^{m+1}\left(A^{k}u\right)={\displaystyle\sum\limits_{k=1}^{m+1}}{\displaystyle\sum\limits_{j=1}^{N}}\lambda_{j}^{k}x_{j}\left(x_{j}^{T}u\right)\\
&  ={\displaystyle\sum\limits_{j=1}^{N}}\left(\lambda_{j}\frac{\lambda_{j}^{m+1}-1}{\lambda_{j}-1}\right)\left(x_{j}^{T}u\right)x_{j} \\ & =\left(\lambda_{1}\frac{\lambda_{1}^{m+1}-1}{\lambda_{1}-1}\right)\left(x_{1}^{T}u\right)x_{1}+{\displaystyle\sum\limits_{j=2}^{N}}\left(\lambda_{j}\frac{\lambda_{j}^{m+1}-1}{\lambda_{j}-1}\right)\left(x_{j}^{T}u\right)x_{j} \\& =\left(\lambda_{1}\frac{\lambda_{1}^{m+1}-1}{\lambda_{1}-1}\right)\left(x_{1}^{T}u\right)x_{1}\left(1+O\left(\sum_{j=2}^{N}\left(\frac{\left|\lambda_{j}\right|}{\left|\lambda_{1}\right|}\right)^{m}\right)\right).
\end{align*}

Literature \cite{VanMieghem2011Graph_spectra} has proved that $x_{1}^{T}u>x_{j}^{T}u$ for all $1<j\leq N$. Accordingly, the term
${\displaystyle\sum\limits_{j=2}^{N}}
\left(\lambda_{j}\frac{\lambda_{j}^{m+1}-1}{\lambda_{j}-1}\right)\left(x_{j}^{T}u\right)x_{j}$
is small in the graphs with a large spectral gap
$(\lambda_{1}-\lambda_{2})$. When $m$ increases,
$d^{(m)}\rightarrow\left(\lambda_{1}\frac{\lambda_{1}^{m+1}-1}{\lambda_{1}-1}\right)\left(x_{1}^{T}u\right)x_{1}$. Moreover,
when $m$ is large, especially when $m\rightarrow\infty$,
$O\left(\sum_{j=2}^{N}\left(\frac{\left|\lambda_{j}\right|}{\left|\lambda_{1}\right|}\right)^{m}\right)\rightarrow
0$ in any graph. Thus we find that $d^{(m)}$ tends to be
proportional to $x_{1}$ when $m$ increases in networks with a large
spectral gap, and $d^{(m)}\sim\lambda_{1}^{(m+1)}(x_{1})$ in
networks when $m\rightarrow\infty$.
\end{proof}

\begin{lemma}
\label{lemma12mass}
In large sparse Erd\H{o}s-R\'{e}nyi (ER) networks, $\rho(D^{(2)},X_{1}
)\geq\rho(D^{(1)},X_{1})\geq\rho(D,X_{1})$.
\end{lemma}

\begin{proof}
see Appendix \ref{proof_app}.
\end{proof}

\section{Application to the inflexible contrarian opinion (ICO) model}
\label{opinionmodel}

In this section we apply the studied centrality metrics to select the
inflexible contrarians in the inflexible contrarian opinion (ICO) model
\cite{li2011strategy} to help one opinion to compete with another. Both
network models and three social networks will be considered.

\subsection{The ICO model} 

The ICO model is a variant of the non-consensus opinion (NCO) model
\cite{shao2009dynamic}. The ICO and NCO models are both opinion
competition models in which two opinions exist and compete with each
other. In the NCO model opinions are randomly assigned to all agents
(nodes). At time $t=0$ each agent is assigned opinion $A$ with a
probability $f$ and opinion $B$ with a probability $1-f$. At each
subsequent time step each agent adopts the opinion of the majority of
its nearest neighbors and itself.  When there is a tie, the opinion of
the agent does not change. All of the updates are made simultaneously in
parallel at each step. The system reaches a state in which the opinions
$A$ and $B$ coexist and are stable when $f$ is above a critical
threshold $f_{c}$.

When the NCO model is in the stable state, the ICO model further selects
a fraction $p_{o}$ of agents with opinion $A$ to be the inflexible
contrarians who will hold opinion $B$, will never change their opinion,
but will influence the opinion of other agents.  The two opinions then
compete with each other according to the update rules of the NCO model.
The system will reach a new stable state by following these opinion
dynamics.

We use $S_{1}$ and $S_{2}$ to denote the size of the largest and the
second largest clusters of agents with opinion $A$ in the new stable
state. A phase transition threshold $f_{c}$ separates two different
phases of the stable state. When $f>f_{c}$, a giant component of agents
with opinion $A$ exists and the coexistence of opinions $A$ and $B$ is
stable. When $f\leq f_{c}$, no giant component of agents with opinion
$A$ exists ($S_{1}=0$). The $f_{c}$ depends on $p_{o}$. When $p_{o}=0$,
the ICO model clearly reduces to the classical NCO model and they have
the same critical threshold $f_{c}$. When $0<p_{o}<p^{\ast}$, the
threshold $f_{c}$ of the ICO model increases with $p_{o}$, but the size
$S_{1}$ for the finial stable state decreases with $p_{o}$. When $p$ is
above a certain value $p^{\ast}$, the phase transition no longer occurs,
and the giant component of agents with opinion $A$ is completely
destroyed ($S_{1}=0$).

\subsection{Strategies of selecting inflexible contrarians using centrality metrics}

The final stable state of the ICO model is affected not only by the
percentage $p_{o}$, but also by how inflexible
contrarian agents are selected. Here we select the inflexible contrarians based on their
centrality metrics. Li \textit{et al.} \cite{li2011strategy} studied the ICO
model by choosing the inflexible contrarian agents with opinion $A$
either randomly or according to highest degree. The degree strategy is significantly
more effective than the random strategy in reducing the size $S_{1}$ of the
largest opinion $A$ cluster in the stable state when $p_{o}$ is the
same. Here we want to determine which centrality metric used to pick the inflexible
contrarians reduces $S_{1}$ most efficiently. We also want to determine 
whether the $S_{1}$ decrease is similar when the inflexible contrarians are
chosen based on two strongly correlated (with a large Pearson correlation
coefficient or a high centrality similarity) centrality metrics.
Here the inflexible contrarians are chosen as nodes with highest (i)
betweenness, (ii) degree, (iii) $1$st-order degree mass, (iv) $2$nd-order
degree mass, (v) eigenvector component, (vi) $k$-shell index, or (vii) leverage
or (viii) chosen randomly.

\begin{figure}[H]
\centering
\includegraphics [width=0.5\textwidth]{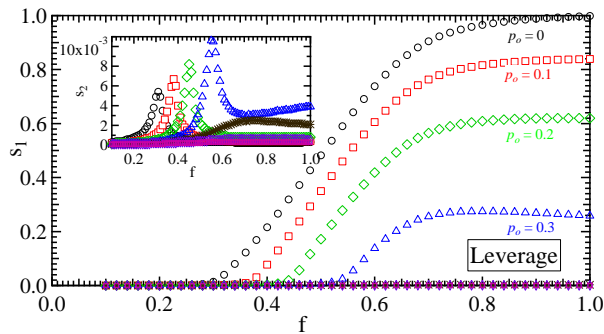}
\caption{An example: the results of leverage strategy. Plot of
  $s_{1}\equiv S_{1} /N$ as a function of $f$ for different values of
  $p_{o} $ for ER networks with $E[D]=4$ and $N=10^{4} $. We denote by
  $S_{1} $ the size of the largest $A$ opinion cluster in the
  steady-state. Different marks show the results of ICO model with
  different $p_{o} $: $p_{o}$=0($\circ$), $p_{o}$=0.1($\square$),
  $p_{o}$=0.2($\diamond $), $p_{o}$=0.3($\triangle$),
  $p_{o}$=0.4($\ast$), $p_{o}$=0.5($\lozenge $),
  $p_{o}$=0.6($\boxtimes$). The insets plot the $s_{2}\equiv S_{2} /N$,
  where $S_{2} $ is the size of the second largest $A$ opinion cluster,
  as a function of the $f$ for different values of $p_{o} $.}
\label{opinion_figure1}
\end{figure}

\subsection{Comparison of inflexible contrarian selection strategies}

We first compare the efficiency in decreasing the size $S_{1}$ of the
largest opinion $A$ cluster in ER and SF networks when choosing the
inflexible contrarians using different centrality metrics. We consider
ER networks ($N=10^{4}$ or $10^{5}$) with $E[D]=4$, and SF networks
($N=10^{4}$ or $10^{5}$) with $\alpha=2.5$, and perform all the
simulations on $10^{3}$ network
realizations. Figure~\ref{opinion_figure1} shows a plot of
$s_{1}=S_{1}/N$ as a function of $f$ for different values of $p_{o}$ in
ER networks (with $N=10^{4}$) using a leverage strategy. The size
$s_{2}=S_{2}/N$ shows a sharp peak, a characteristic of a second-order
phase transition, in the insets of Fig.~\ref{opinion_figure1}. As
$p_{o}$ increases, $f_{c}$ shifts to a larger value and the largest
cluster becomes significantly smaller. When $p>p^{\ast}$, the giant
component with opinion $A$ disappears, i.e., $S_{1}=0$.  For example,
the $p^{\ast}$ value for the leverage strategy is between 0.3 and 0.4
(see Fig.~\ref{opinion_figure1}).  A small $p^{\ast}$ implies that the
inflexible contrarians can efficiently destroy the largest opinion $A$
cluster. We can compare the efficiency of the strategies in decreasing
$S_{1}$ by the value of $p^{\ast}$.  When we compare strategies in the
ICO model with the same $p_{o}$, a larger phase transition $f_{c}$ for a
strategy indicates that the inflexible contrarians chosen using this
strategy decreases $S_{1}$ more efficiently.  Figure~\ref{function}
plots the phase transition $f_{c}$ as a function of $p_{o}$.  Note that
the efficiency of each strategy is ranked in decreasing order as:
Leverage, Degree, Betweenness, $1$st-order Degree mass, $2$nd-order
Degree mass, $k$-shell index, Principal Eigenvector, and Random. The
same result can be also found in ER and SF networks with $N=10^{5}$.

\begin{figure}[H]
\centering
\subfloat[]{\label{function}
\includegraphics[width=0.5\textwidth]{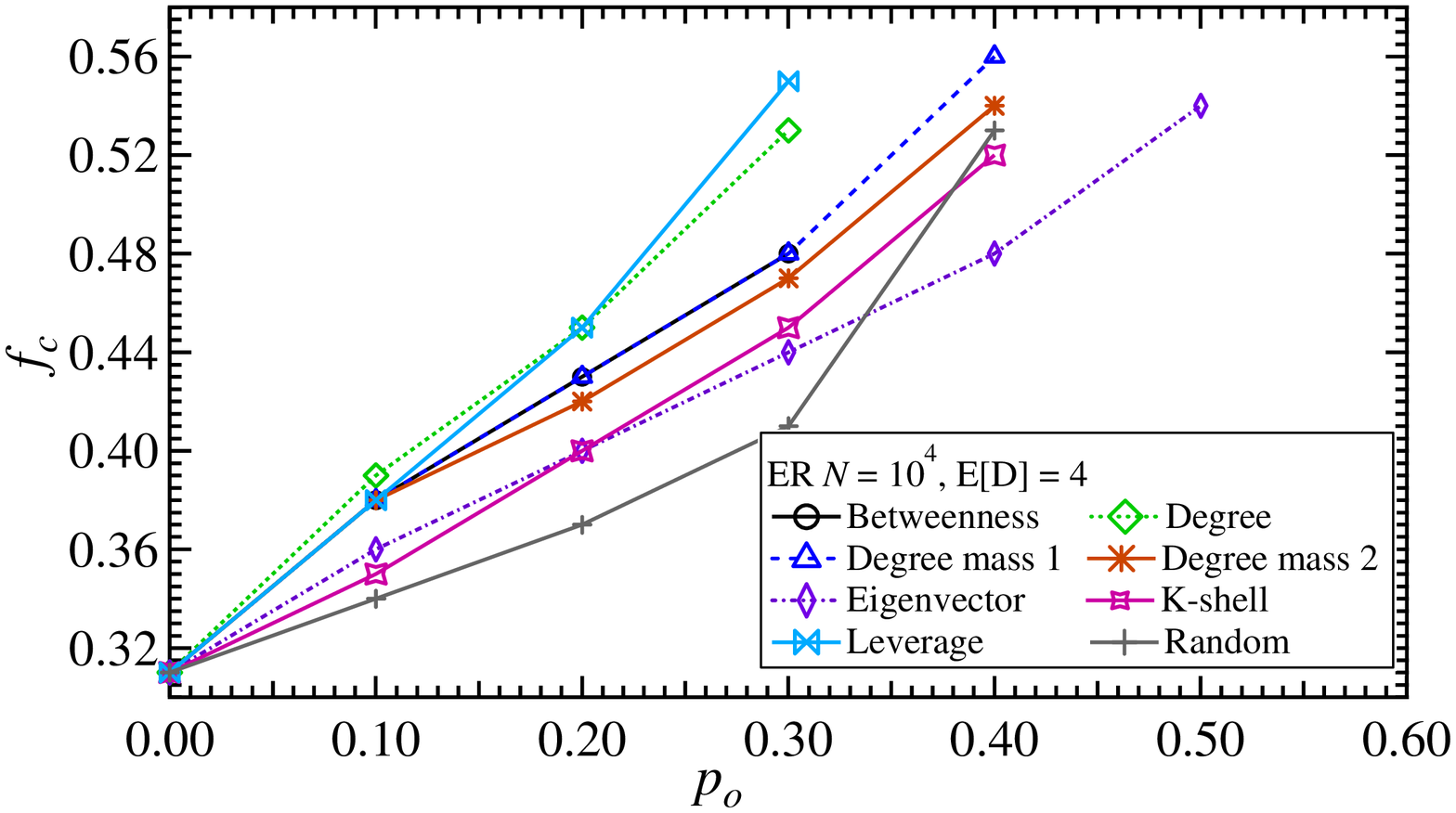}}
\subfloat[]{\label{SFvesus}
\includegraphics[width=0.5\textwidth]{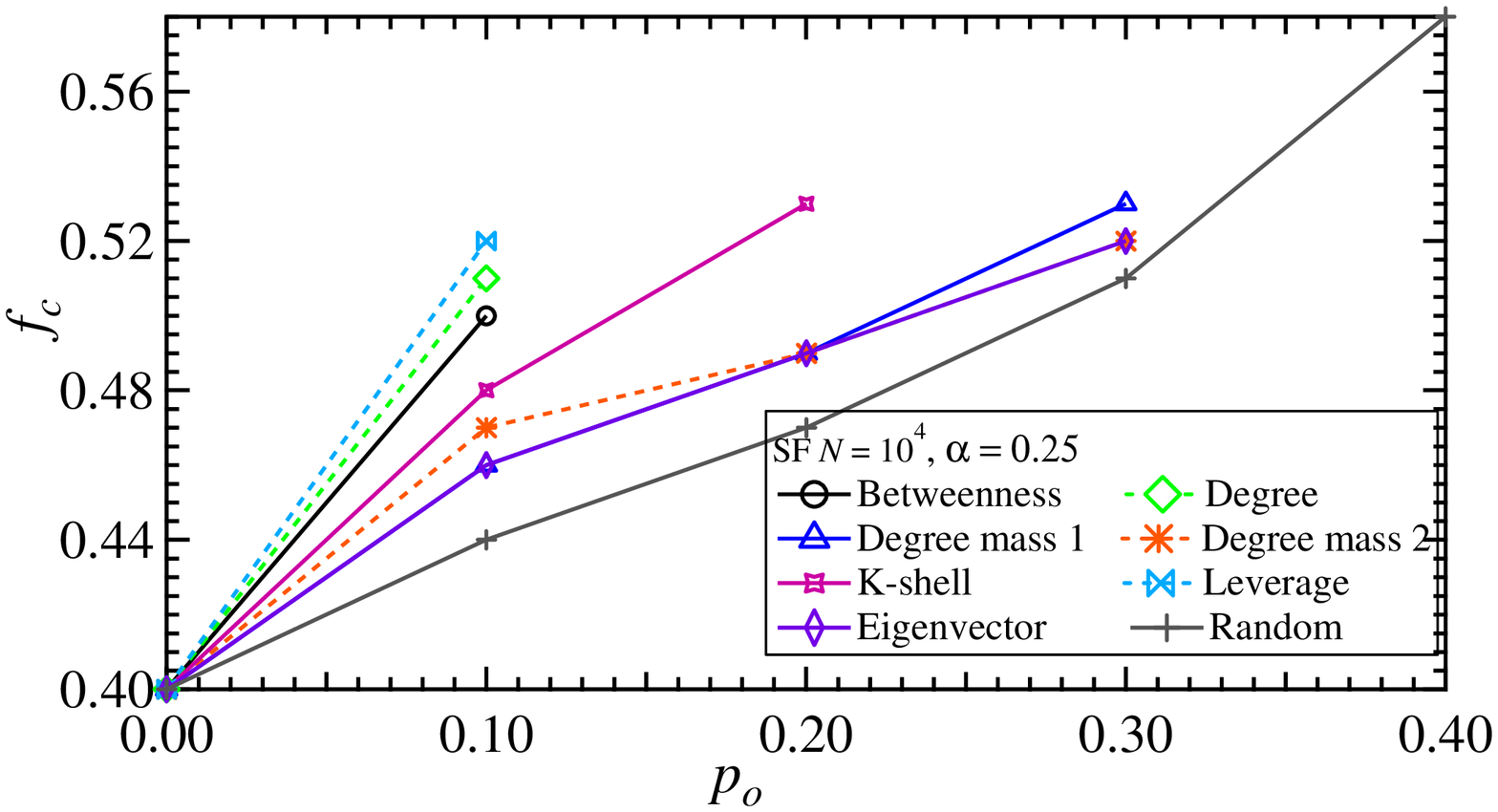}}
\bigskip
\caption{Plot of $f_{c}$ as a function of $p_{o}$ for strategies $1$ to
  $8$: (a) in ER graphs with $N=10^{4}$, $E[D]=4$; (b) in SF graphs with
  $N=10^{4}$, $D_{\min}=2$, $\alpha=2.5$.}
\label{opinion_models_10000}
\end{figure}

We find that all strategies are more efficient in SF networks than in ER
networks of the same size. We base this on two observations. First, the
relative change of $f_{c}$ with $p_{o}$ for all strategies in SF
networks is larger than it is in ER networks. Second, the $p^{\ast}$ for
all strategies in SF is much smaller than it is in ER networks. The
reason for this may be that (i) hubs can be readily selected as
inflexible contrarians when using centrality metrics in SF networks, and
(ii) hubs can strongly influence the opinion of their large number of
nearest neighbors.

Figure \ref{opinin_realworld} compares these centrality metrics in
real-world networks, i.e., the ConMat 95-99 network, the ConMat 95-03
network, and the Astro\_Ph network. Note that the inflexible contrarians
selected using the leverage $L_{n}$, the betweenness $B_{n}$, and the
degree $D$ are the most efficient in helping opinion $B$ win the
competition. The similar behaviors of the three strategies are supported
by the large Pearson correlation coefficient $\rho(B_{n}$, $D)$ and the
large centrality similarities $M_{B_{n},D}(\Upsilon)$,
$M_{L_{n},D}(\Upsilon)$ and $M_{L_{n},B_{n}}(\Upsilon)$.

\begin{figure}[H]
\centering
\subfloat[]{\label{conmat99}
\includegraphics[width=0.333\textwidth]{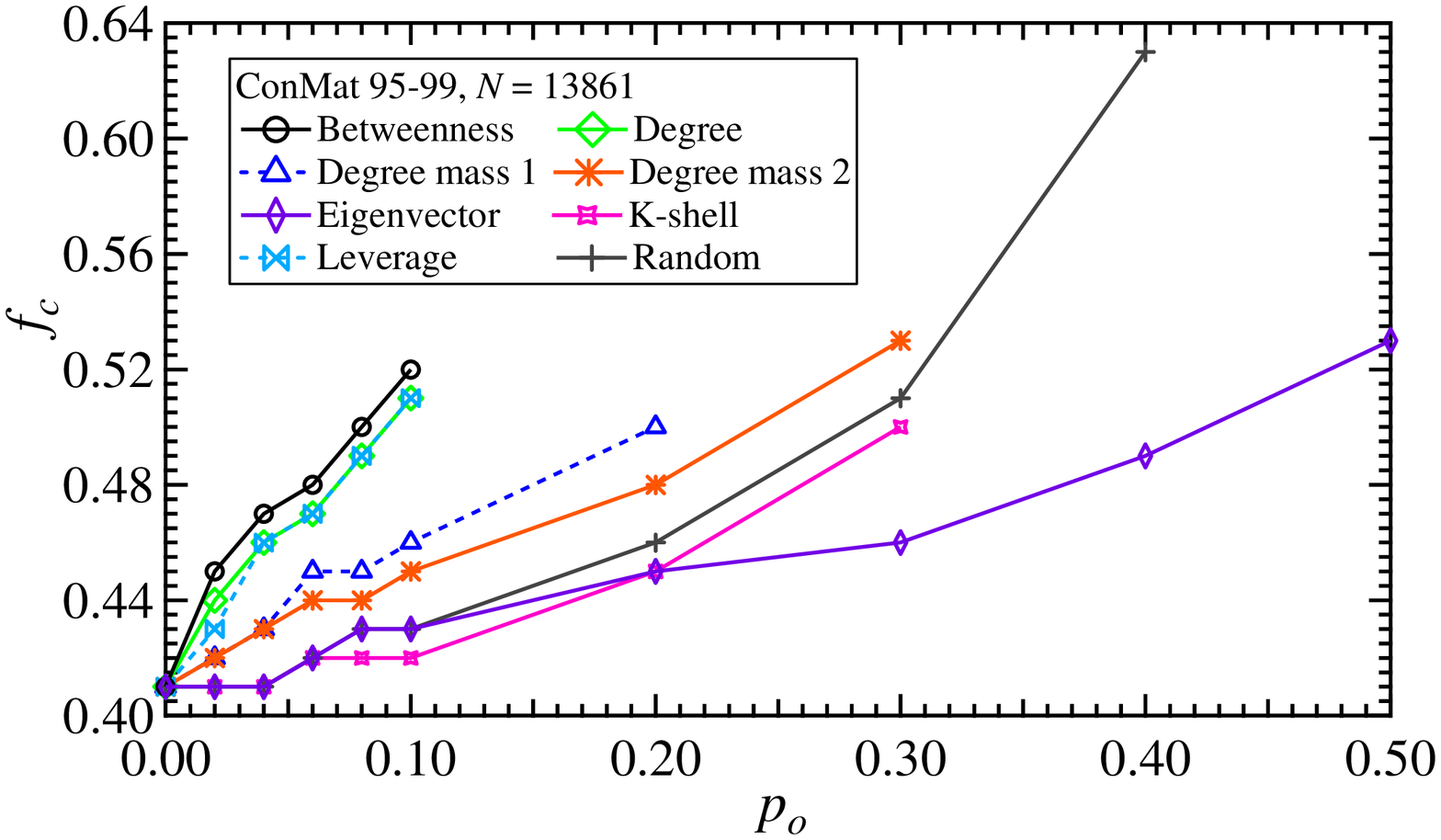}}
\subfloat[]{\label{conmat03}%
\includegraphics[width=0.333\textwidth]{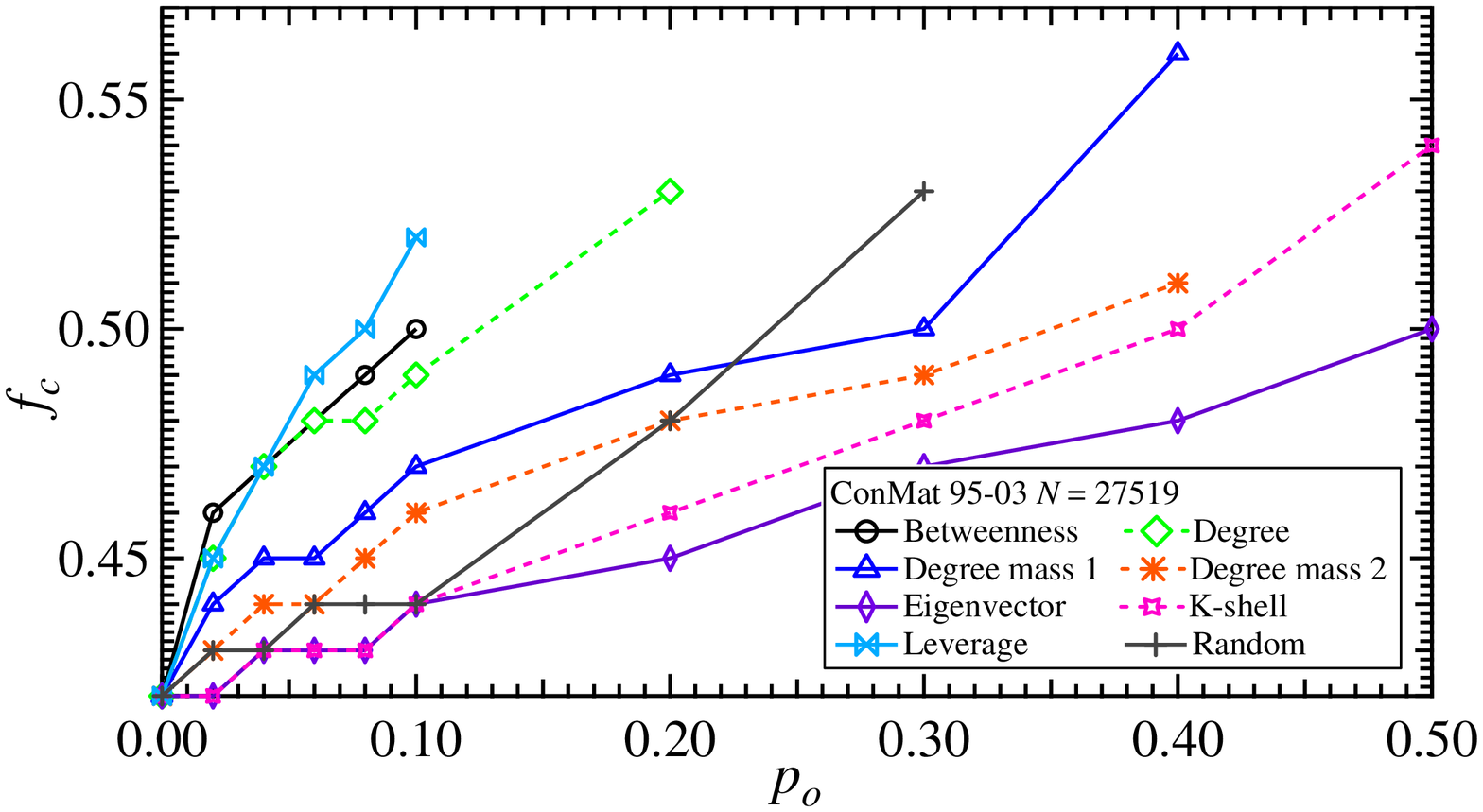}}
\subfloat[]{\label{astroph}%
\includegraphics[width=0.333\textwidth]{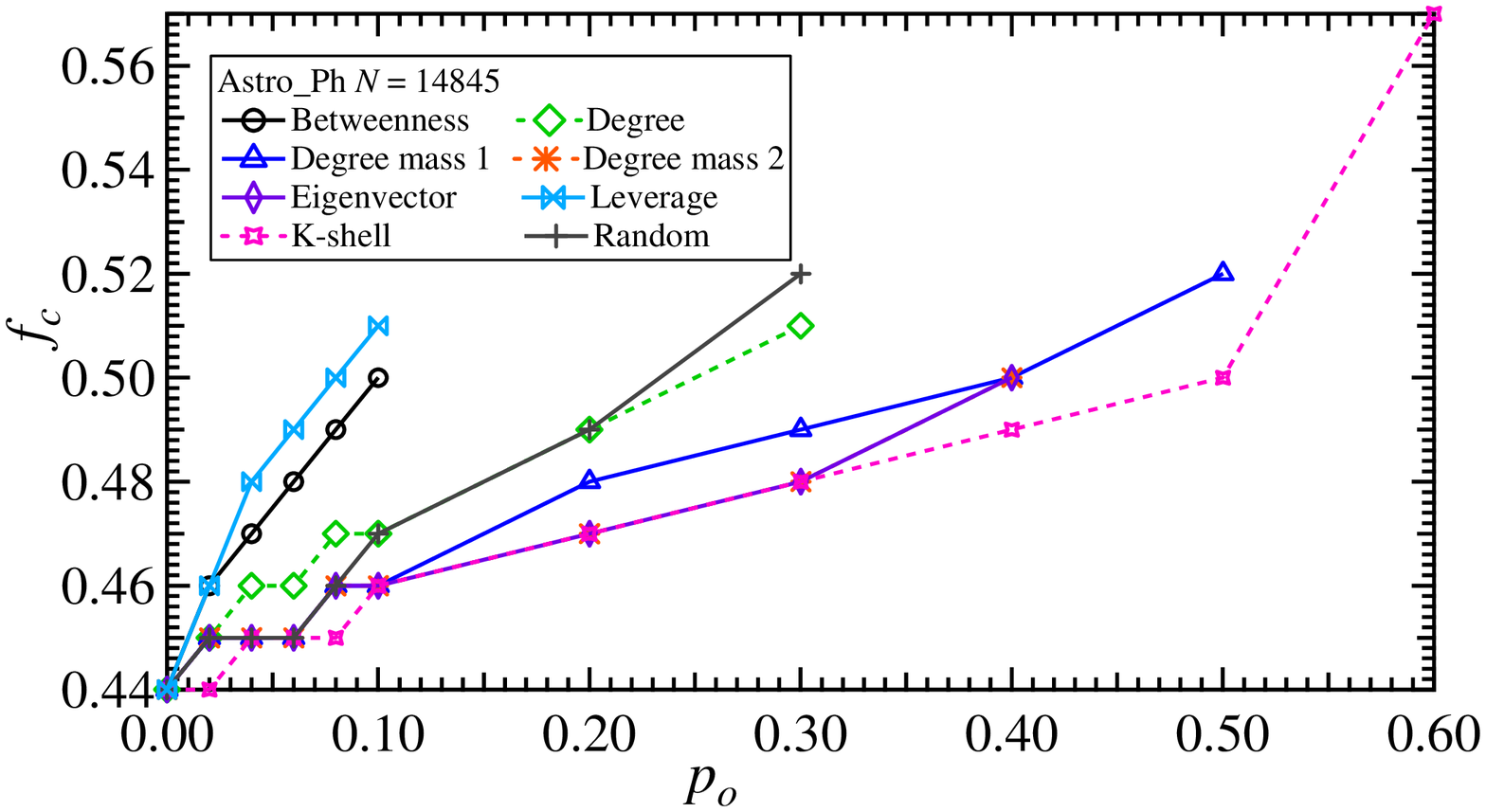}}
\bigskip
\caption{Plot of $f_{c}$ as a function of $p_{o}$ for strategies in
  social networks: (a) in network of coauthorships between scientists
  posting preprints on ConMat E-Print Archives between $1995$ to $1999$;
  (b) in network of coauthorships between scientists posting preprints
  on ConMat E-Print Archives between $1995$ to $2003$; (c) in network of
  coauthorships between scientists posting preprints on Astrophysics
  E-Print Archives between $1995$ to $1999$.}
\label{opinin_realworld}
\end{figure}

In both network models and real-world networks, strongly correlated
centrality metrics tend to perform similarly. For example, we have
discovered both numerically and theoretically that
$\rho(D^{(2)},X_{1})\geq\rho(D^{(1)},X_{1})$. Correspondingly, the
principal eigenvector $x_{1}$ strategy performs closer to the 2nd-order
degree mass $D^{(2)}$ than the 1st-order degree mass $D^{(1)}$
in the ICO model.

\section{Conclusion}

In this paper we have studied the correlation between widely studied and
recently proposed centrality metrics in numerous real-world networks as
well as in network models, i.e., as in Erd\H{o}s-R\'{e}nyi (ER) random
networks and scale-free (SF) networks. A strong correlation between
two centrality metrics indicates the possibility of approximating one
centrality metric, usually the one with a higher computational
complexity, using the other. We study the correlations between the
centrality metrics using the Pearson correlation coefficient and the
centrality similarity. An important finding is that the degree $D$, the
$1$st-order degree mass $D^{(1)}$, and the $2$nd-order degree mass
$D^{(2)}$ are strongly correlated with the betweenness $B_{n}$, the
closeness $C_{n}$, and the principal eigenvector $x_{1}$,
respectively. This observation is partially supported by our analytical
proof that $\rho(X_{1},D^{(2)})>\rho(X_{1},D^{(1)})>\rho(X_{1},D)$. 

We have introduced the degree mass $D^{(m)}$ as a new network centrality
metric. The $0$th-order degree mass is the degree and the high-order
($m\rightarrow\infty$) degree mass is proportional to the principal
eigenvector $x_{1}$. We also find that the influence of network size
(the number $N$ of nodes) on the Pearson correlation coefficients is
small. In addition, the leverage $L_{n}$ has high centrality
similarities with the degree $D$ and the betweenness $B_{n}$. We use
these centrality metrics to select the inflexible contrarians in the ICO
model to help one opinion to compete with the other. The leverage
$L_{n}$ turns out to be the most efficient strategy in both network
models and real-world networks. We also find that strongly correlated
metrics perform similarly in the ICO model. This suggests that the
metrics with a low computational complexity, such as the degree $D$ and
the leverage $L_{n}$, could be used to approximate more complex metrics,
e.g., the betweenness $B_{n}$, to locate important nodes in complex
networks. Examples of important nodes would include inflexible
contrarians in opinion propagation networks and nodes that should be
immunized in disease transmission networks.

\textbf{Acknowledgements}

The authors are grateful to Shlomo Havlin for discussion and useful
comments. This work has been supported by the European Commission within
the framework of the CONGAS project FP7-ICT-2011-8-317672 and the China
Scholarship Council (CSC).

\bibliographystyle{ieeetr} 
\bibliography{Opinionref}

\appendix

\bigskip
\section{Description of the real-world networks}
\label{Appendix_discription of realworld networks}

\bigskip

\subsection{Descriptions}
\begin {table}[H]
\caption{Descriptions of real-world networks.}
\bigskip
\label{Tablerealworlddescriptions}%
\centering
{\footnotesize
\begin{tabular}{|c|c|p{13cm}|}
\hline

Index & Networks & Descriptions \\\hline\hline
1 & American airline & The direct airport-to-airport American mileage a
maintained by the U.S. Bureau of Transportation Statistics.\\\hline

2 & American football &This is the network of American football games between
Division IA colleges during regular season Fall 2000, as compiled by M. Girvan
and M. Newman.\\\hline

3 & ARPANET80 & The Advanced Research Projects Agency Network as seen in 1980.\\\hline

4 & Celegensneural & Network representing the neural network of C. Elegans.\\\hline

5 & Dophins & An undirected social network of frequent associations between 62
dolphins in a community living off Doubtful Sound, New Zealand.\\\hline

6 & Dutch soccer & Dutch football players represent the nodes. Two nodes are
linked if they played together a match.\\\hline

7 & Gnutella 1  & Gnutella snapshots. Four different crawls are available.\\
8 & Gnutella 2  &\ \\
9 & Gnutella 3  &\ \\
10 & Gnutella 4 &\                   \\\hline

11 & Karate &  Social network of friendships between 35 members of a karate club
at a US university in the 1970.\\\hline

12 & LesMis & Coappearance network of characters in the novel Les Miserables.\\\hline

13 & Surfnet&  SURFNET topology inferred from the switch interface interconnections.\\\hline

14 & Electric s208 & ISCAS89 Sequential Benchmark Circuits. Each node represents
a logical operation implemented \\
15 & Electric s420& physically. Links between them relate their inputs/outputs.\\
16 & Electric s838&\ \\\hline

17& Epowergridl1 &Power-grid infrastructure at three different levels of
one city-area in Western Europe.\\
18& Epowergridl2 &\ \\
19& Epowergridl3 &\ \\\hline
20& Erailwayl1& Railway infrastructure at two levels of one
Western-European country\\
21& Erailwayl2&\ \\\hline

22& WordAdj&Adjacency network of common adjectives and nouns in the novel
David Copperfield by Charles Dickens.\\\hline

23& WordAdjEnglish& Word-adjacency networks of texts in English, French and
Japanese separately.\\
24&WordAdjFranch&\ \\
25&WordAdjJapanese&\ \\\hline

26&Internet AS (01')& Internet snapshot retrieved from the merge of different
data sources (BGP routing tables and updates: Route Views, RIPE, Abilene,
CERNET, BGP View).\\\hline
27& Astro\_Ph&Network of coauthorships between scientists posting preprints
on the Astrophysics E-Print Archive between Jan 1, 1995 and December 31, 1999.\\\hline

28&SciMet&Web of Science C. The citation network was created using the Web
of Science database SciMet. Networks created with the tool HistCite.\\\hline

29& HighE-th&High Energy Theory C. Network of coauthorships between
scientists posting preprints on the High-Energy Theory E-Print Archive between
Jan 1, 1995 and December 31, 1999.\\\hline

30&CondMat 95-03& Network of coauthorships between scientists
posting preprints on the Condensed Matter E-Print 
\\
31&CondMat 95-99 & Archive. We have two networks corresponding to different periods of time. Periods are Jan 1, 1995-December 31, 1999 and 2003 respectively.\\\hline
32& Dutch Roadmap& A graph representing the interconnection between cities in
the Netherlands.\\\hline

33& Network Science C& Coauthorship network of scientists working on network
theory and experiment, as compiled by M. Newman in May 2006.\\\hline

34& Next Generation& A typical Next Generation Transport network.\\\hline

\end{tabular}
}
\end{table}
\bigskip

\subsection{Properties of the real-world networks}

The properties of real-world networks are shown in the Table
\ref{Tablerealworldproperty}. The definition of these properties has been
described in detail in \cite{li2011correlation}.
\begin{table}[H]
\caption{Properties of real-world networks. The real-world network index is shown in Table \ref{Tablerealworlddescriptions}. $N$ is the number of nodes, $L$
is the number of links. $E[H]$ is the average shortest path, $C_{G}$ is the
clustering coefficient of networks. $\rho_{D}$ is the degree correlation
coefficient (called the assortativity) of networks. $\lambda_{1}$ is the
largest eigenvalue (called spectral radius) of the adjacency matrix of the
network. $\mu_{N-1}$ is the second smallest Laplace eigenvalue (called
spectral radius) of the networks. $\mu_{1}/\mu_{N-1}$ is the ratio of the
largest eigenvalue $\mu_{1}$ and the second smallest eigenvalue $\mu_{1}$ of
Laplacian matrix.\ $R_{G}$\ is the effective graph resistance.}

\bigskip
\label{Tablerealworldproperty}%
\ \ \ \ \ \ \ \ \ \ \ \ \ \ \ \ \ \ \ \ \ \ \ \ \ \ \ \ \ \ \ \ \ \ \ \ \ \ \ \ \ \ \ \ \ \ \ \ \ \ \ \ \ \ \ \ \ \centering{\tiny $%
\begin{array}
[c]{|ccccccccccccc|}\hline\hline
\text{Index} & N & L & E[H] & C_{G} & \rho_{D} & \lambda_{1} & \mu_{N-1} &
\mu_{1}/\mu_{N-1} & R_{G} & E[D] & \sqrt{Var[D]} & H_{\max}\\\hline
1 & 2179 & 31326 & 3.0262 & 0.4849 & -0.0409 &
144.6112 & 0.2082 & 2.0675e^{3} & 1.6072e^{4} & 28.7526 & 56.6782 & 8\\
2 & 115 & 613 & 2.5082 & 0.4032 & 0.1624 &
10.7806 & 1.4590 & 10.7350 & 1.5086e^{3} & 10.6609 & 0.8835 & 4\\
3 & 71 & 86 & 6.4849 & 0.0141 & -0.2613 & 2.7648 & 0.0374 &
170.2063 & 7.0158e^{3} & 2.4225 & 0.7442 & 17\\
4 & 297 & 2148 & 2.4553 & 0.2924 & -0.1632 & 24.3655 &
0.8485 & 159.1562 & 1.3710e^{4} & 14.4646 & 12.9443 & 5\\
 5 & 62 & 159 & 3.3570 & 0.2590 & -0.0436 & 7.1936 & 0.1730 &
78.7034 & 1.8643e^{3} & 5.1290 & 2.9319 & 8\\
6 & 685 & 10310 & 4.4583 & 0.7506 & -0.0634 & 50.8428 &
0.1613 & 372.0373 & 3.1157e^{4} & 30.1022 & 21.1957 & 11\\
 7 & 737 & 803 & 9.1351 & 0.0063 & -0.1934 & 4.8913 &
0.0073 & 2.6292e^{3} & 1.4181e^{6} & 2.1791 & 2.0069 & 24\\
 8 & 1568 & 1906 & 6.1037 & 0.0192 & -0.0946 & 13.7828 &
0.0167 & 1.1205e^{4} & 4.0212e^{4} & 2.4311 & 5.5778 & 21\\
 9 & 435 & 459 & 6.7085 & 0.0145 & -0.3301 & 8.2281 &
0.0110 & 5.9278e^{3} & 4.2533e^{5} & 2.1103 & 5.1534 & 20\\
 10 & 653 & 738 & 5.4513 & 0.0232 & -0.2459 & 12.1145 &
0.0231 & 6.2319e^{3} & 6.6603e^{5} & 2.2603 & 7.0228 & 15\\
11 & 35 & 134 & 1.9126 & 0.3908 & -0.5036 & 9.6253 & 1.7264 &
12.6030 & 221.6283 & 7.6571 & 4.7265 & 3\\
 12 & 77 & 254 & 2.6411 & 0.5731 & -0.1652 & 12.0058 & 0.2050 &
180.9490 & 3.0166e^{3} & 6.5974 & 6.0006 & 5\\
 13 & 65 & 111 & 4.1236 & 0.0359 & 0.2288 & 5.0523 & 0.1137 &
92.7068 & 3.2979e^{3} & 3.4154 & 1.9046 & 10\\
14 & 122 & 189 & 4.9278 & 0.0591 & -0.0020 & 4.1036 &
0.0836 & 135.2786 & 1.3082e^{4} & 3.0984 & 1.4395 & 11\\
15 & 252 & 399 & 5.8064 & 0.0651 & -0.0059 & 4.3600 &
0.0512 & 297.3970 & 5.8313e^{4} & 3.1667 & 1.5340 & 13\\
16 & 512 & 819 & 6.8585 & 0.0547 & -0.0300 & 5.0097 &
0.0285 & 809.9553 & 2.5149e^{5} & 3.1992 & 1.6296 & 15\\
17 & 3419 & 3953 & 21.1147 & 0.0120 & -0.1283 & 5.1781 &
<e^{-5} & >e^{15} & 4.8953e^{7} & 2.3124 & 1.8425 & 51\\
18 & 1205 & 1384 & 12.3547 & 0.0171 & 0.1082 & 4.8994 &
0.0022 & 9.1191e^{3} & 4.3901e^{6} & 2.2971 & 1.3609 & 31\\
19 & 395 & 441 & 13.6088 & 0.0201 & -0.0235 & 4.4854 &
0.0020 & 8.8844e^{3} & 7.2535e^{5} & 2.2329 & 1.2834 & 42\\
20 & 8710 & 11332 & 79.0448 & 0.0212 & -0.0219 & 2.9865 &
<e^{-5} & >e^{15} & 7.2107e^{8} & 2.6021 & 0.7696 & 213\\
21 & 689 & 778 & 34.1261 & 0.0731 & 0.0980 & 3.6926 &
7.7321e^{-3} & 1.0526e^{4} & 3.9229e^{6} & 2.2583 & 0.7658 & 84\\
22 & 112 & 425 & 2.5356 & 0.1728 & -0.1293 & 13.1502 &
0.6950 & 72.0767 & 3.7941e^{3} & 7.5893 & 6.8512 & 5\\
23 & 7377 & 44205 & 2.7780 & 0.4085 & -0.2366 &
109.4416 & <e^{-5} & 9.1266e^{15} & 2.2149e^{7} & 11.9846 & 60.8260 & 8\\
24 & 8308 & 23832 & 3.2189 & 0.2138 & -0.2330 &
60.6735 & 0.1197 & 1.5810e^{4} & 3.9917e^{7} & 5.7371 & 34.8979 & 9\\
25 & 2698 & 7995 & 3.0771 & 0.2196 & -0.2590 &
42.9980 & <e^{-5} & 5.8851e^{15} & 4.3489e^{6} & 5.9266 & 24.6695 & 8\\
26 & 12254 & 25319 & 3.6214 & 0.2992 & -0.1903 & 61.1066 &
<e^{-5} & 4.8974e^{15} & 1.0349e^{8} & 4.1324 & 33.5463 & 11\\
27 & 14845 & 119652 & 4.7980 & 0.6696 & 0.2277 & 73.8868 &
0.0302 & 1.1966e^{4} & 7.2012e^{7} & 16.1202 & 21.7466 & 14\\
28 & 2678 & 10368 & 4.1797 & 0.1736 & -0.0352 & 20.4290 &
0.0853 & 1.9365e^{3} & 2.9549e^{6} & 7.7431 & 9.2480 & 12\\
29 & 5835 & 13815 & 7.0264 & 0.5062 & 0.1852 & 18.0442 &
0.0214 & 2.3870e^{3} & 2.8800e^{7} & 4.7352 & 4.5571 & 19\\
30 & 27519 & 116181 & 5.7667 & 0.6546 & 0.1657 &
40.3097 & 0.0276 & 7.3675e^{3} & 3.3638e^{8} & 8.4437 & 10.8110 & 16\\
31 & 13861 & 44619 & 6.6278 & 0.6514 & 0.1571 &
24.9822 & 0.0292 & 3.6992e^{3} & 1.1613e^{8} & 6.4381 & 6.7598 & 18\\
32 & 29663 & 34982 & 148.7102 & 0.0443 & 0.2462 & 3.4567 &
<e^{-5} & >e^{15} & 1.5472e^{10} & 2.3586 & 0.6823 & 531\\
33 & 379 & 914 & 6.0419 & 0.7412 & -0.0817 & 10.3755 &
0.0152 & 2.3053e^{3} & 1.4826e^{5} & 4.8232 & 3.9272 & 17\\
34 & 29902 & 32707 & 7109.8681 & 0.0306 & -0.0355 &
49.5455 & <e^{-5} & >e^{15} & 2.1188e^{12} & 2.1876 & 9.7574 &
14253\\\hline\hline
\end{array}
$}\end{table}%

\bigskip
\section{Pearson correlation coefficients between centrality metrics}
\label{Appen_Pearson}
\bigskip

The correlation indexes mentioned in the following images and tables are the indexes for pairs of centrality metrics: $1$. $(B_{n},C_{n})$; $2$. $(B_{n},D)$; $3$. $(B_{n},x_{1})$; $4$. $(B_{n},K_{s})$; $5 $. $(B_{n},L_{n})$; $6$%
. $(B_{n},D^{(1)})$; $7$. $(B_{n},D^{(2)})$; $8$. $(C_{n},D)$; $9$. $(C_{n}%
,x_{1})$; $10$. $(C_{n},K_{s})$; $11$. $(C_{n},L_{n})$; $12$. $(C_{n},D^{(1)})$;
$13$. $(C_{n},D^{(2)})$; $14$. $(D,x_{1})$; $15$. $(D,K_{s})$; $16$. $(D,L_{n})$;
$17$. $(D,D^{(1)})$; $18$. $(D,D^{(2)})$; $19$. $(x_{1},K_{s})$; $20$%
. $(x_{1},L_{n})$; $21$. $(x_{1},D^{(1)})$; $22$. $(x_{1},D^{(2)})$; $23$%
. $(K_{s},L_{n})$; $24$. $(K_{s},D^{(1)})$; $25$. $(K_{s},D^{(2)})$; $26$%
. $(L_{n},D^{(1)})$; $27$. $(L_{n},D^{(2)})$; $28$. $(D^{(1)},D^{(2)})$.

\begin{figure}[H]%
\centering
\includegraphics[
trim=0.023782in 0.000000in -0.023783in 0.000000in,
height=2.5538in,
width=5.3714in
]%
{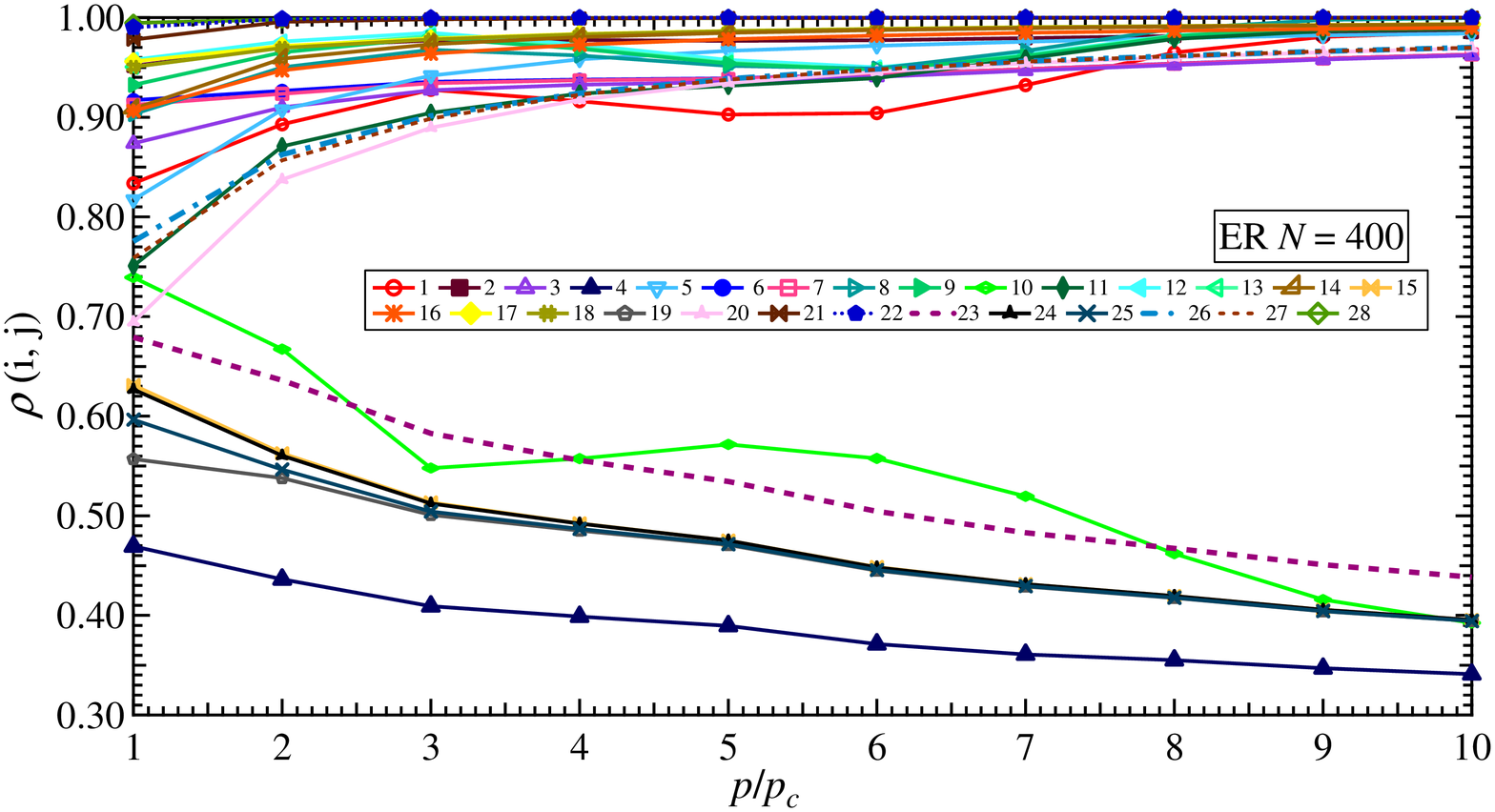}%
\caption{Pearson correlation coefficient between any two centrality metrics as
a function of the link density $p$, in ER networks ($N=400$). The number in
the annotation is the correlation index.}%
\label{ER3}%
\end{figure}
\begin{figure}[H]%
\centering
\includegraphics[
height=2.8366in,
width=5.4457in
]%
{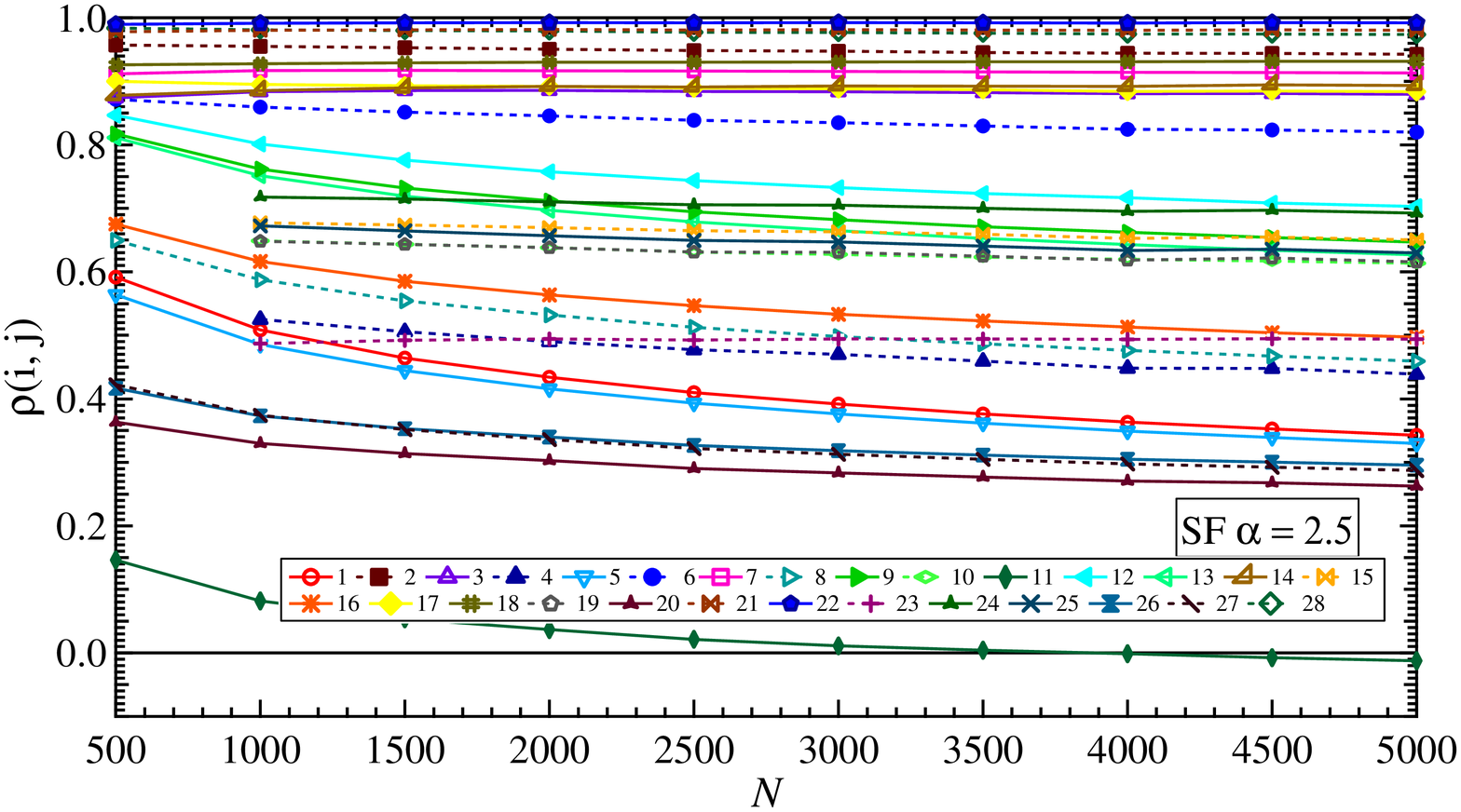}%
\caption{Pearson correlation coefficient between any two centrality metrics as
a function of the size $N$ of networks, in scale-free networks ($\alpha=2.5$).
The number in the annotation is the correlation index. }%
\label{SF}%
\end{figure}

\bigskip

\begin{table}[H]
\caption{Pearson correlation coefficients among the centrality metrics in
the real-world networks.}%
\bigskip
\label{corr1}%
\ \ \ \ \ \ \ \ \ \ \ \ \ \ \ \ \ \ \ \ \ \ \ \ \ \ \ \ \ \ \ \ \ \ \ \ \ \ \ \ \ \ \ \centering{\tiny $%
\begin{array}
[c]{|cccccccccc|}\hline\hline
\text{Index} & 1 & 2 & 3 & 4 & 5 & 6 & 7 & 8 & 9\\
& \rho(B_{n},C_{n}) & \rho(B_{n},D) & \rho(B_{n},x_{1}) & \rho(B_{n}%
,K_{s}) & \rho(B_{n},L_{n}) & \rho(B_{n},D^{(1)}) & \rho(B_{n},D^{(2)}) &
\rho(C_{n},D) & \rho(C_{n},x_{1})\\\hline
1 & 0.3667 & 0.5690 & 0.4119 & 0.3377 & 0.4027 &
0.4314 & 0.4224 & 0.7580 & 0.7684\\
2 & 0.8167 & 0.2813 & 0.1450 & 0.0871 & 0.3212 &
0.2230 & 0.2075 & 0.2913 & 0.2462\\
3 & 0.7129 & 0.7235 & 0.5358 & 0.3496 & 0.55585 & 0.7660 &
0.7593 & 0.4308 & 0.6851\\
 4 & 0.4271 & 0.7805 & 0.5206 & 0.1822 & 0.4212 &
0.5388 & 0.6044 & 0.6997 & 0.7827\\
5 & 0.6657 & 0.5902 & 0.2835 & 0.4703 & 0.5639 & 0.5131 &
0.4850 & 0.7127 & 0.6979\\
6 & 0.3303 & 0.4909 & 0.0857 & 0.1523 & 0.4170 &
0.3807 & 0.3113 & 0.2701 & -0.1604\\
7 & 0.4456 & 0.7292 & 0.4780 & 0.5182 & 0.4556 & 0.7575 &
0.7882 & 0.3973 & 0.5241\\
 8 & 0.2196 & 0.9691 & 0.7006 & 0.2677 & 0.2679 & 0.3858 &
0.9416 & 0.2225 & 0.5469\\
 9 & 0.2475 & 0.8839 & 0.4926 & 0.4667 & 0.4356 & 0.3533 &
0.8283 & 0.1763 & 0.5112\\
 10 & 0.2338 & 0.9603 & 0.5848 & 0.3296 & 0.3880 & 0.2640 &
0.8839 & 0.1774 & 0.5733\\
 11 & 0.8699 & 0.9651 & 0.8757 & 0.3782 & 0.8707 & 0.7999 &
0.9166 & 0.8853 & 0.9599\\
 12 & 0.6287 & 0.7468 & 0.4231 & 0.2388 & 0.5317 & 0.5534 &
0.5468 & 0.7997 & 0.6812\\
 13 & 0.7136 & 0.8743 & 0.7365 & 0.6345 & 0.6985 & 0.7999 &
0.7816 & 0.8290 & 0.9286\\
14 & 0.6408 & 0.7475 & 0.5595 & 0.2147 & 0.5551 &
0.7357 & 0.7227 & 0.6127 & 0.7987\\
15 & 0.5956 & 0.6933 & 0.5514 & 0.1583 & 0.4508 &
0.7084 & 0.7203 & 0.5541 & 0.7178\\
16 & 0.5323 & 0.7044 & 0.5410 & 0.1314 & 0.3913 &
0.6971 & 0.7661 & 0.4623 & 0.5633\\
17 & 0.2349 & 0.3843 & 0.1180 & 0.1189 & 0.1889 &
0.4101 & 0.4082 & 0.1082 & 0.0607\\
 18 & 0.3210 & 0.7005 & 0.5517 & 0.0560 & 0.2686 &
0.6772 & 0.7144 & 0.2946 & 0.4627\\
19 & 0.3001 & 0.7081 & 0.4775 & 0.1060 & 0.2945 &
0.6371 & 0.6825 & 0.2395 & 0.4925\\
20 & 0.2664 & 0.1565 & -0.0442 & 0.1979 & 0.1112 &
0.1805 & 0.1876 & 0.1477 & 0.0209\\
 21 & 0.5022 & 0.3274 & 0.0364 & 0.3836 & 0.2548 & 0.2790 &
0.2540 & 0.2428 & 0.1141\\
22 & 0.6559 & 0.9150 & 0.8226 & 0.3517 & 0.6586 & 0.7891 &
0.8444 & 0.8410 & 0.9245\\
23 & 0.1880 & 0.9225 & 0.6525 & 0.2068 & 0.2642 &
0.4157 & 0.7765 & 0.3535 & 0.6528\\
24 & 0.1874 & 0.9714 & 0.8047 & 0.2729 & 0.2636 &
0.4403 & 0.9385 & 0.2625 & 0.6215\\
25 & 0.2747 & 0.9660 & 0.7859 & 0.3249 & 0.3584 &
0.5266 & 0.8972 & 0.3868 & 0.6880\\
26 & 0.1382 & 0.9826 & 0.7994 & 0.3292 & 0.2290 & 0.3441 &
0.9582 & 0.1631 & 0.5776\\
27 & 0.3764 & 0.6787 & 0.4353 & 0.2869 & 0.4631 & 0.5670 &
0.5270 & 0.6109 & 0.4220\\
28 & 0.4068 & 0.8185 & 0.6959 & 0.3147 & 0.4401 & 0.7143 &
0.7605 & 0.6741 & 0.7030\\
29 & 0.4526 & 0.7798 & -0.0109 & 0.3574 & 0.5079 & 0.6700 &
0.5803 & 0.5774 & 0.0119\\
30 & 0.3801 & 0.7534 & 0.3753 & 0.3152 & 0.4488 &
0.5933 & 0.5173 & 0.5989 & 0.3906\\
31 & 0.4002 & 0.7225 & 0.2781 & 0.2607 & 0.4581 &
0.5718 & 0.4816 & 0.5616 & 0.3248\\
32 & 0.2214 & 0.1741 & -0.0037 & 0.1619 & 0.1117 &
0.1719 & 0.1608 & 0.1450 & -0.0221\\
33 & 0.4302 & 0.6883 & 0.1884 & 0.1917 & 0.4707 & 0.5630 &
0.4997 & 0.3468 & 0.2593\\
34 & -0.1342 & -0.0436 & -0.6295 & -0.9718 & 0.9538 &
-0.9051 & -0.1342 & 0.0313 & 0.2446\\\hline\hline
\end{array}
$}\end{table}

\begin{table}[H]
\ \ \ \ \ \ \ \ \ \ \ \ \ \ \ \ \ \ \ \ \ \ \ \ \ \ \ \ \ \ \ \ \ \ \ \ \ \ \ \ \ \ \ \ \ \ \ \ \ \centering{\tiny $%
\begin{array}
[c]{|ccccccccccc|}\hline\hline
\text{Index} & 10 & 11 & 12 & 13 & 14 & 15 & 16 & 17 & 18 & 19\\
& \rho(C_{n},K_{s}) & \rho(C_{n},L_{n}) & \rho(C_{n},D^{(1)}) & \rho
(C_{n},D^{(2)}) & \rho(D,x_{1}) & \rho(D,K_{s}) & \rho(D,L_{n}) &
\rho(D,D^{(1)}) & \rho(D,D^{(2)}) & \rho(x_{1},K_{s})\\\hline
1 & 0.8174 & 0.5944 & 0.7903 & 0.7712 & 0.9592 & 0.8730 & 0.7259 & 0.9657 &
0.9643 & 0.9254\\
2 & 0.1742 & 0.2704 & 0.2826 & 0.2839 & 0.7501 & 0.3881 & 0.9181 & 0.9619 &
0.9314 & 0.2456\\
3 & 0.3807 & 0.2524 & 0.5598 & 0.5870 & 0.4650 & 0.5127 & 0.8914 & 0.9020 &
0.9079 & 0.1326\\
4 & 0.6861 & 0.5776 & 0.8680 & 0.7951 & 0.7810 & 0.5434 & 0.7886 & 0.8830 &
0.9311 & 0.5572\\
5 & 0.7498 & 0.6094 & 0.7475 & 0.7422 & 0.7196 & 0.8303 & 0.9050 & 0.9574 &
0.9417 & 0.5388\\
6 & 0.0680 & 0.2221 & 0.2381 & 0.1801 & 0.6237 & 0.7300 & 0.8963 & 0.9393 &
0.8801 & 0.7983\\
7 & 0.5073 & 0.2052 & 0.6184 & 0.6248 & 0.4660 & 0.5933 & 0.8117 & 0.8217 &
0.8573 & 0.3912\\
8 & 0.4015 & 0.0017 & 0.7515 & 0.3210 & 0.6523 & 0.3463 & 0.3888 & 0.3594 &
0.9132 & 0.1840\\
9 & 0.2377 & -0.3534 & 0.8326 & 0.3544 & 0.5811 & 0.3316 & 0.4651 & 0.3050 &
0.9493 & 0.2032\\
10 & 0.2234 & -0.2234 & 0.8594 & 0.2967 & 0.6366 & 0.2492 & 0.3751 & 0.2256 &
0.9481 & 0.0868\\
11 & 0.5492 & 0.7227 & 0.9606 & 0.9463 & 0.9392 & 0.5331 & 0.9390 & 0.8718 &
0.9714 & 0.6221\\
12 & 0.5622 & 0.6340 & 0.8375 & 0.7931 & 0.8467 & 0.7969 & 0.8474 & 0.9455 &
0.9380 & 0.8100\\
13 & 0.7311 & 0.3466 & 0.9330 & 0.9363 & 0.9046 & 0.8289 & 0.7598 & 0.9486 &
0.9391 & 0.8425\\
14 & 0.5670 & 0.3265 & 0.7388 & 0.7574 & 0.6757 & 0.4296 & 0.8184 & 0.9260 &
0.9225 & 0.3108\\
15 & 0.5257 & 0.2675 & 0.6964 & 0.7100 & 0.6147 & 0.3995 & 0.7980 & 0.9078 &
0.9200 & 0.2464\\
16 & 0.4949 & 0.1937 & 0.6534 & 0.6411 & 0.4120 & 0.3738 & 0.7690 & 0.8670 &
0.9055 & 0.1143\\
17 & -0.0402 & -0.0122 & 0.1651 & 0.1653 & 0.2143 & 0.4102 & 0.6878 & 0.7733 &
0.8456 & 0.0447\\
18 & 0.1582 & 0.1027 & 0.4752 & 0.4902 & 0.5040 & 0.1904 & 0.5901 & 0.8725 &
0.8851 & 0.0638\\
19 & 0.2490 & 0.2137 & 0.5599 & 0.5316 & 0.5183 & 0.2287 & 0.5911 & 0.7611 &
0.8338 & 0.0327\\
20 & 0.1649 & 0.0836 & 0.1829 & 0.2016 & 0.1031 & 0.7905 & 0.9247 &
0.9522 & 0.9241 & 0.1132\\
21 & 0.4880 & 0.0325 & 0.3314 & 0.3382 & 0.2678 & 0.4149 & 0.7508 & 0.8884 &
0.8524 & 0.0835\\
22 & 0.8194 & 0.7371 & 0.9451 & 0.9123 & 0.9575 & 0.6433 & 0.8327 & 0.9390 &
0.9707 & 0.7010\\
23 & 0.7195 & 0.3891 & 0.8312 & 0.5353 & 0.8704 & 0.4649 & 0.4862 &
0.6580 & 0.9504 & 0.7992\\
24 & 0.6355 & 0.0669 & 0.8167 & 0.4111 & 0.8733 & 0.4146 & 0.3627 &
0.5403 & 0.9779 & 0.6980\\
25 & 0.6814 & 0.2080 & 0.8410 & 0.5506 & 0.8911 & 0.5155 & 0.5048 & 0.6631 &
0.9694 & 0.7628\\
26 & 0.4291 & -0.0707 & 0.7971 & 0.2788 & 0.8253 & 0.3935 & 0.2696 &
0.3771 & 0.9754 & 0.5413\\
27 & 0.5427 & 0.2819 & 0.5861 & 0.5264 & 0.7188 & 0.8070 & 0.5920 &
0.9352 & 0.8728 & 0.5695\\
28 & 0.8188 & 0.5093 & 0.7923 & 0.7456 & 0.8345 & 0.6962 & 0.7237 &
0.9204 & 0.9236 & 0.6212\\
29 & 0.4884 & 0.2103 & 0.6517 & 0.6022 & 0.1789 & 0.7311 & 0.7080 &
0.9080 & 0.8292 & 0.5171\\
30 & 0.6341 & 0.2404 & 0.6153 & 0.5392 & 0.6346 & 0.7339 & 0.6197 &
0.9035 & 0.8259 & 0.5001\\
31 & 0.5157 & 0.2077 & 0.6067 & 0.5300 & 0.5304 & 0.7166 & 0.6631 &
0.8941 & 0.8021 & 0.4229\\
32 & 0.1465 & -0.0170 & 0.2033 & 0.2220 & 0.0364 & 0.5291 & 0.7674 &
0.9271 & 0.8880 & 0.0101\\
33 & 0.0926 & 0.0970 & 0.4562 & 0.4120 & 0.4748 & 0.6803 & 0.7723 &
0.8795 & 0.8415 & 0.4195\\
34 & 0.3609 & -0.3531 & 0.3378 & 0.0649 & 0.7297 & 0.0866 & 0.0487 &
0.1570 & 0.9858 & 0.6768\\\hline\hline
\end{array}
$}\end{table}\begin{table}[H]
\ \ \ \ \ \ \ \ \ \ \ \ \ \ \ \ \ \ \ \ \ \ \ \ \ \ \ \ \ \ \ \ \ \ \ \ \ \ \ \ \ \ \ \ \ \ \ \ \ \ \ \centering{\tiny $%
\begin{array}
[c]{|cccccccccc|}\hline\hline
\text{Index} & 20 & 21 & 22 & 23 & 24 & 25 & 26 & 27 & 28\\
& \rho(x_{1},L_{n}) & \rho(x_{1},D^{(1)}) & \rho(x_{1},D^{(2)}) & \rho
(K_{s},L_{n}) & \rho(K_{s},D^{(1)}) & \rho(K_{s},D^{(2)}) & \rho(L_{n}%
,D^{(1)}) & \rho(L_{n},D^{(2)}) & \rho(D^{(1)},D^{(2)})\\\hline
1 & 0.6327 & 0.9978 & 0.9998 & 0.7122 & 0.9389 & 0.9245 & 0.6604 & 0.6405 &
0.9984\\
2 & 0.4881 & 0.8660 & 0.9134 & 0.4481 & 0.3467 & 0.3274 & 0.7771 & 0.7189 &
0.9929\\
3 & 0.1934 & 0.6460 & 0.7101 & 0.5798 & 0.4773 & 0.4407 & 0.6485 & 0.6530 &
0.9811\\
4 & 0.6130 & 0.9783 & 0.9885 & 0.7710 & 0.6277 & 0.5737 & 0.6605 & 0.6789 &
0.9813\\
5 & 0.4991 & 0.8285 & 0.8842 & 0.8506 & 0.8171 & 0.7668 & 0.7887 & 0.7535 &
0.9913\\
6 & 0.3684 & 0.8132 & 0.8867 & 0.6089 & 0.8563 & 0.8700 & 0.7478 & 0.6517 &
0.9864\\
7 & 0.2262 & 0.6736 & 0.7412 & 0.4906 & 0.6480 & 0.5920 & 0.5024 & 0.4922 &
0.9475\\
8 & 0 & 0.8135 & 0.8463 & 0.5030 & 0.3187 & 0.2061 & -0.0050 &
0.1176 & 0.4936\\
9 & -0.1161 & 0.7007 & 0.7440 & 0.3782 & 0.2365 & 0.2762 & -0.3636 & 0.2134 &
0.4889\\
10 & -0.1437 & 0.7414 & 0.8018 & 0.5184 & 0.1290 & 0.1398 & -0.3598 & 0.1438 &
0.3751\\
11 & 0.8128 & 0.9837 & 0.9930 & 0.5722 & 0.6484 & 0.5928 & 0.7290 & 0.8623 &
0.9568\\
12 & 0.6520 & 0.9427 & 0.9691 & 0.7984 & 0.8524 & 0.8447 & 0.7713 & 0.7455 &
0.9924\\
13 & 0.4673 & 0.9841 & 0.9927 & 0.6005 & 0.8604 & 0.8512 & 0.5510 & 0.5248 &
0.9969\\
14 & 0.3310 & 0.8087 & 0.8589 & 0.2983 & 0.4885 & 0.4576 & 0.6007 & 0.5809 &
0.9839\\
15 & 0.2497 & 0.7503 & 0.8010 & 0.2684 & 0.4523 & 0.4137 & 0.5520 & 0.5475 &
0.9788\\
16 & 0.0562 & 0.5789 & 0.6656 & 0.2530 & 0.4262 & 0.3673 & 0.4862 & 0.4847 &
0.9533\\
17 & 0.0545 & 0.3371 & 0.3805 & 0.7626 & 0.1513 & 0.1393 & 0.2283 &
0.2760 & 0.9458\\
18 & 0.0501 & 0.6365 & 0.6794 & 0.3420 & 0.1429 & 0.1199 & 0.2204 & 0.2123 &
0.9812\\
19 & 0.0748 & 0.7433 & 0.7697 & 0.3351 & 0.1335 & 0.1077 & 0.0448 & 0.1010 &
0.9619\\
20 & 0.0564 & 0.1303 & 0.1454 & 0.6233 & 0.8541 & 0.8624 & 0.7665 &
0.7184 & 0.9907\\
21 & 0.0347 & 0.4062 & 0.4780 & 0.2918 & 0.4205 & 0.3829 & 0.4013 & 0.3398 &
0.9842\\
22 & 0.7490 & 0.9949 & 0.9983 & 0.8031 & 0.7300 & 0.6910 & 0.7541 & 0.7622 &
0.9888\\
23 & 0.6646 & 0.9320 & 0.9790 & 0.7406 & 0.8912 & 0.6890 & 0.6611 &
0.6156 & 0.8432\\
24 & 0.3912 & 0.8774 & 0.9476 & 0.5641 & 0.7939 & 0.5488 & 0.3408 &
0.3734 & 0.6794\\
25 & 0.5507 & 0.9242 & 0.9721 & 0.6990 & 0.8180 & 0.6646 & 0.4857 & 0.5386 &
0.8112\\
26 & 0.1486 & 0.8169 & 0.8977 & 0.4876 & 0.5646 & 0.4417 & 0.0699 &
0.1943 & 0.4845\\
27 & 0.2248 & 0.8789 & 0.9367 & 0.4761 & 0.7840 & 0.7124 & 0.3996 &
0.3245 & 0.9845\\
28 & 0.4680 & 0.9417 & 0.9682 & 0.7181 & 0.7457 & 0.6886 & 0.5866 &
0.5501 & 0.9877\\
29 & 0.0427 & 0.2885 & 0.3822 & 0.5164 & 0.7657 & 0.7361 & 0.4493 &
0.3477 & 0.9771\\
30 & 0.1765 & 0.8431 & 0.9205 & 0.5016 & 0.7344 & 0.6617 & 0.3726 &
0.2850 & 0.9795\\
31 & 0.1358 & 0.7641 & 0.8725 & 0.4877 & 0.7372 & 0.6597 & 0.3945 &
0.2903 & 0.9731\\
32 & 0.0063 & 0.0524 & 0.0629 & 0.3943 & 0.5167 & 0.4740 & 0.4892 &
0.4156 & 0.9878\\
33 & 0.1267 & 0.7062 & 0.8105 & 0.5701 & 0.7390 & 0.6966 & 0.5089 &
0.4324 & 0.9766\\
34 & -0.5920 & 0.7797 & 0.8022 & -0.9766 & 0.9347 & 0.1797 & -0.9156 &
-0.0611 & 0.2549\\\hline\hline
\end{array}
$}\end{table}

\bigskip

\section{Proof of Lemmas}\label{proof_app}
\subsection{}\begin{lemma}\label{lemma_proof}
In an Erd\H{o}s-R\'{e}nyi (ER) random network $G_{p}(N)$, when $N\rightarrow
\infty$, the average $1$st-order degree mass is
\begin{equation}
E[D^{(1)}]=N(2p+p^{2}N)\text{,} \label{average_D1}%
\end{equation}
and the variance is
\begin{equation}
Var[D^{(1)}]=N(2p+4p^{2}N+p^{3}N^{2})\text{.} \label{Var_D1}%
\end{equation}
The average and the variance of $2$nd-order degree mass are
\begin{equation}
E[D^{(2)}]=N(2p+3p^{2}N+p^{3}N^{2})\text{,} \label{average_D2}%
\end{equation}
\begin{equation}
Var[D^{(2)}]=N(2p+14p^{2}N+17p^{3}N^{2}+7p^{4}N^{3}+p^{5}N^{4})\text{.}\label{Var_D2}%
\end{equation}

\end{lemma}

\begin{proof} 
The generating function for the probability distribution of node degree is
defined as%
\[
\varphi_{D}(z)=\sum_{k=0}^{N-1}z^{k}\mathrm{Prob}[D=k]\text{,}%
\]
and the generating function of the degree of the node that we arrive at by
following a randomly chosen link is%
\begin{equation}
\frac{\sum_{k}k\mathrm{Prob}[D=k]z^{k}}{\sum_{k}k\mathrm{Prob}[D=k]}%
=z\frac{\varphi_{D}^{\prime}(z)}{E[D]}\text{,}
\label{pdf_node_arrived_by_a_random_link}%
\end{equation}
where $E[.]$ is the expectation. If we start at a randomly chosen node, the
generating function of the degree of a nearest neighbor of this node follows
Eq. (\ref{pdf_node_arrived_by_a_random_link}). The $1$st-order degree mass
$D^{(1)}$\ of a node equals the degree sum of the node and its neighbors. The
generating function has the `powers' property \cite{newman2001random}, that the
distribution of the $1$st-order degree mass of a node obtained from one
nearest neighbor is generated by%
\[
\varphi_{D}(z)^{\ast}=z^{2}\frac{\varphi_{D}^{\prime}(z)}{E[D]}\text{,}%
\]
then, the distribution of the total of the $1$st-order degree mass over
$k$\ independent realizations ($k$\ nearest neighbors) of the node is
generated by $k$th power of $\varphi_{D}(z)^{\ast}$ as%
\begin{equation}
\varphi_{D^{(1)}}(z)=\varphi_{D}(\varphi_{D}(z)^{\ast})=\sum_{k}%
\mathrm{Prob}[D=k]\left(  z^{2}\frac{\varphi_{D}^{\prime}(z)}{E[D]}\right)
^{k}\text{.} \label{generating_function_D1_method}%
\end{equation}
For ER networks, $E[D]=(N-1)p$ is the average\ degree in an ER network
$G_{p}(N)$, and $\varphi_{D}(z)=(1-p+pz)^{N-1}$, thus,%
\begin{equation}
\varphi_{D^{(1)}}(z)=((1-p)+z^{2}p(1-p+pz)^{N-2})^{N-1}\text{,}
\label{generating_function_D1}%
\end{equation}
In addition, the generating function has\ the `Moments' property \cite{newman2001random}, that $E[\left(  D^{(1)}\right)  ^{n}]=\left[
(z\frac{d}{dz})^{n}\varphi_{D^{(1)}}(z)\right]  _{z=1}$. Together with
$Var[D^{(1)}]=E[\left(  D^{(1)}\right)  ^{2}]-E[D^{(1)}]^{2}$, we arrive at
the (\ref{average_D1}) and (\ref{Var_D1}), when $N\rightarrow\infty$.

Similarly, the distribution of the $2$nd-order degree mass is generated by
$\varphi_{D}(\varphi_{D^{(1)}}(\varphi_{D^{(1)}}(z)))$. Hence, we obtain the
generating function of the $2$nd-order degree mass as%
\[
\varphi_{D^{(2)}}(z)=(1-p+pz^{2}(1-p+pz)^{N-2}(1-p+pz^{2}(1-p+pz)^{N-2}%
)^{N-2})^{N-1}\text{,}%
\]
Thus, we can obtain (\ref{average_D2}) and (\ref{Var_D2}).
\end{proof}

\bigskip
\subsection{Proof of Lemma\ref{lemma12mass}}

\begin{proof}
The eigenvalue equation $Ax=\lambda x$ leads to $\lambda_{1}^{k}x_{1}=A^{k}x_{1}$, from which we obtain $u^{T}x_{1}%
{\displaystyle\sum\limits_{j=1}^{m}}
\lambda_{1}^{j}=u^{T}\left(
{\displaystyle\sum\limits_{j=1}^{m}}
A^{j}\right)  x_{1}\text{,}$ where $u^{T}x_{1}=NE[X_{1}]$ and $u^{T}%
{\displaystyle\sum\limits_{j=1}^{m+1}}
A^{j}=\left(  d^{(m)}\right)  ^{T}$. Hence, the relation between the principal
eigenvector and the $m$th-order degree mass vector can be expressed as %
$E[X_{1}]N%
{\displaystyle\sum\limits_{j=1}^{m+1}}
\lambda_{1}^{j}=\left(  d^{(m)}\right)  ^{T}x_{1}$\text{,} leading to
\begin{equation}
E[D^{(m)}X_{1}]=E[X_{1}]{\displaystyle\sum\limits_{j=1}^{m+1}}\lambda_{1}^{j}.
\end{equation}

The Pearson correlation coefficient follows as%
\begin{equation}
\rho(D^{(m)},X_{1})=\frac{E[D^{(m)}X_{1}]-E[D^{(m)}]E[X_{1}]}{\sqrt
{Var[D^{(m)}]}\sqrt{Var[X_{1}]}}=\frac{\left(
{\displaystyle\sum\limits_{j=1}^{m+1}}
\lambda_{1}^{j}-E[D^{(m)}]\right)  E[X_{1}]}{\sqrt{Var[D^{(m)}]}%
\sqrt{Var[X_{1}]}}\text{.} \label{correlationcoef}%
\end{equation}
The ratio of the two Pearson correlation coefficients is%
\begin{align}
\frac{\rho(D^{(1)},X_{1})}{\rho(D,X_{1})}=\frac{\sqrt{Var[D]}}{\sqrt
{Var[D^{(1)}]}}\left(  1+\frac{(\lambda_{1}^{2}-E[D^{2}])}{\left(  \lambda
_{1}-E[D]\right)  }\right)\label{ratio1and0}
\end{align}
For large ER graphs, $E[D]=(N-1)p\rightarrow Np$, $E[D^{2}%
]=(N-1)^{2}p^{2}-(N-1)p^{2}+(N-1)p\rightarrow N^{2}p^{2}-Np^{2}+Np$ and
$Var[D]=(N-1)p(1-p)\rightarrow Np(1-p)$. From (\ref{Var_D1}), we obtain%
\begin{align}
\frac{\sqrt{Var[D]}}{\sqrt{Var[D^{(1)}]}} =\sqrt{\frac{(1-p)}{(E[D]+2)^{2}-2}} >\frac{1}{E[D]+2}\text{.}\label{inequality_1}
\end{align}
When $N\rightarrow\infty$ and $Np=\varsigma$ ($\varsigma$ is a constant and independent of $N$), the spectral radius $\lambda_{1}\rightarrow \varsigma$, in sparse
random graphs \cite{krivelevich2003largest, farkas2001spectra}. With (\ref{ratio1and0}) and (\ref{inequality_1}), $\rho(D^{(1)},X_{1})\geq\rho(D,X_{1})$ is proved.

The ratio of the two Pearson correlation coefficients is%
\[
\frac{\rho(D^{(2)},X_{1})}{\rho(D^{(1)},X_{1})}=\frac{\left(  \lambda
_{1}+\lambda_{1}^{2}+\lambda_{1}^{3}-E[D^{(2)}]\right)  \sqrt{Var[D^{(1)}]}%
}{\left(  \lambda_{1}+\lambda_{1}^{2}-E[D^{2}]-E[D]\right)  \sqrt
{Var[D^{(2)}]}}\text{,}%
\]
with (\ref{average_D2}) and $\lambda_{1}\rightarrow Np$, when $N\rightarrow\infty$ we arrive at
\[
\frac{\left(  \lambda_{1}+\lambda_{1}^{2}+\lambda_{1}^{3}-E[D^{(2)}]\right)
}{\left(  \lambda_{1}+\lambda_{1}^{2}-E[D^{2}]-E[D]\right)  }=2E[D]+1\text{.}%
\]
With (\ref{Var_D1}) and (\ref{Var_D2}), for large sparse random networks, $\rho(D^{(2)},X_{1}
)\geq\rho(D^{(1)},X_{1})$ is proved.
\end{proof}
\end{document}